\documentclass[pdflatex,sn-basic]{sn-jnl}
\catcode`\|=12\relax 

\usepackage[algo2e,linesnumbered,algoruled,vlined]{algorithm2e}
\usepackage{longtable}
\usepackage{hyperref}
\usepackage{etoolbox}
\usepackage[final]{changes}
\definechangesauthor[name={Alex}, color=purple]{A}
\definechangesauthor[name={Daniel}, color=green]{D}

\jyear{2022}%

\theoremstyle{thmstyleone}%
\newtheorem{theorem}{Theorem}
\newtheorem{corollary}{Corollary}%
\newtheorem{lemma}{Lemma}

\theoremstyle{thmstyletwo}%

\theoremstyle{thmstylethree}%
\newtheorem{definition}{Definition}%

\AtBeginEnvironment{algorithm2e}{\fontsize{8}{8}\linespread{1.3}\selectfont}

\begin{document}

\title{Subnetwork Constraints for Tighter Upper Bounds and Exact Solution of the Clique Partitioning Problem}



\author*[1,2,3]{\fnm{Alexander} \sur{Belyi}}\email{bely@math.muni.cz}
\author[1,4,5]{\fnm{Stanislav} \sur{Sobolevsky}}\email{ss9872@nyu.edu}
\author[3]{\fnm{Alexander} \sur{Kurbatski}}\email{kurbatski@bsu.by}
\author[6]{\fnm{Carlo} \sur{Ratti}}\email{ratti@mit.edu}

\affil*[1]{\orgdiv{Department of Mathematics and Statistics}, \orgname{Faculty of Science, Masaryk University}, \orgaddress{\street{Kotlarska~2}, \city{Brno}, \postcode{611\,37}, \country{Czech Republic}}}

\affil[2]{\orgdiv{Senseable City Lab}, \orgname{FM IRG, SMART Centre}, \orgaddress{\street{1~Create~Way}, \postcode{138602}, \country{Singapore}}}

\affil[3]{\orgdiv{Faculty of Applied Mathematics and Computer Science}, \orgname{Belarusian State University}, \orgaddress{\street{4 Nezavisimosti Avenue}, \city{Minsk}, \postcode{220030}, \country{Belarus}}}

\affil[4]{\orgdiv{Institute of Law and Technology}, \orgname{Faculty of Law, Masaryk University}, \orgaddress{\street{Veveri 70}, \city{Brno}, \postcode{611\,80}, \country{Czech Republic}}}

\affil[5]{\orgdiv{Center for Urban Science and Progress}, \orgname{New York University}, \orgaddress{\street{370~Jay~Street}, \city{Brooklyn}, \postcode{11201}, \state{NY}, \country{USA}}}

\affil[6]{\orgdiv{Senseable City Lab}, \orgname{Massachusetts Institute of Technology}, \orgaddress{\street{77~Massachusetts~Avenue}, \city{Cambridge}, \postcode{02139}, \state{MA}, \country{USA}}}

\abstract{
We consider a variant of the clustering problem for a complete weighted graph.
The aim is to partition the nodes into clusters maximizing the sum of the edge weights within the clusters.
This problem is known as the clique partitioning problem, being NP-hard in the general case of having edge weights of different signs.
We propose a new method of estimating an upper bound of the objective function that we combine with the classical branch-and-bound technique to find the exact solution.
We evaluate our approach on a broad range of random graphs and real-world networks.
The proposed approach provided tighter upper bounds and achieved significant convergence speed improvements compared to known alternative methods.
}



\maketitle

\section{Introduction\label{sect:Intro}}

Clustering is one of the fundamental problems in data analysis and machine learning~\citep{Clustering1999}.
In general terms, clustering means grouping similar objects together.
At the same time, many real-world systems could be represented as networks~\citep{newman2018networks} or graphs.
Traditionally, \added[id=A]{the} word ``graph'' refers to the mathematical model of the underlying network, but we will use these terms interchangeably as synonyms.
Graphs are a powerful mathematical model often used to study a broad range of objects and their relations.
So, the clustering of real-world objects is often modeled and formulated mathematically as the clustering of vertices of a graph.
If it is possible to quantify the similarity between objects, one can construct a complete graph where vertices correspond to the objects, and edge weights represent their similarity.
In this case, the clustering problem could be formulated as a clique partitioning problem~\citep{Grotschel1989cutting,Grotschel1990facets}.

Formally, given a complete weighted graph, the clique partitioning problem~(CPP) is to find such a partition of vertices into groups (or clusters or modules) that maximizes the sum of weights of edges connecting vertices within the same groups.
Obviously, this problem is not trivial only when the graph has both positive and negative edge weights.
In the literature, this problem is known under different names, including clique partitioning, correlation clustering, and signed graph clustering~\citep{hausberger2022distributed}.

In a more general case, when a system could still be represented as a network of connections between nodes, but similarities between objects are not given, the clustering problem spawned a separate field of research known as community detection in networks~\citep{fortunato2010}.
There are many approaches to community detection, but one of the most widely adopted is to define a similarity or a strength of the connection between nodes and then optimize the sum of these strengths within clusters.
Probably the most well-known quality function of such partitioning is modularity~\citep{Girvan2002PNAS,Newman2004,Newman2006PNAS}.
For each pair of nodes, a modularity score is defined as a normalized difference between actual edge weight and expected weight in a random graph that preserves node degrees.
Modularity of a partition is then just a sum of modularity scores of pairs of nodes placed in the same cluster.
The problem of finding an optimal partition in terms of modularity can now be formulated as the clique partitioning problem in a graph whose edge weights correspond to modularity scores.

There are many real-world applications of CPP.
The most famous, including those studied in original works by~\cite{Grotschel1989cutting}, come from biology, group technology~\citep{Oosten2001,Wang2006}, and transportation~\citep{dorndorf2008modelling}.
Community detection done through modularity maximization solved as CPP could apply to areas ranging from geo-informatics~\citep{belyi2016flickr,belyi2017IJGIS} and tourism management~\citep{Xu2021KoreanTourists} to biochemistry~\citep{Guimera2005FunctionalCartography} and the study of social networks~\citep{Girvan2002PNAS}.
The practical usefulness of the problem continues to attract researchers' attention.
However, solving CPP is hard.

NP-hardness of CPP has been known since~\cite{wakabayashi1986PhD}.
And the same result was later proven for modularity maximization too~\citep{Brandes2008}.
Thus, most of the scholars' efforts have been aimed at developing heuristic approaches that allow finding relatively good solutions relatively quickly.
Among such approaches were simulated annealing and tabu search~\citep{deAmorim1992,gao2022improvingSA}, ejection chain and Kernighan-Lin heuristic~\citep{Dorndorf1994}, noising method~\citep{Charon2006}, neighborhood search~\citep{brusco2009neighborhood,Brimberg2017}, iterative tabu search~\citep{palubeckis2014ITS}, and their combinations~\citep{Zhou2016CPP-P3}.
Usually, graphs considered in operational research are not too big, comprising hundreds to a few thousands of nodes, and the quality of approximate solutions is high~\citep{Zhou2016CPP-P3,hu2020two-model,lu2021hybrid}.
At the same time, in network science, graphs could be extremely large, spanning over millions of nodes.
Therefore, an extensive search for the solutions close to optimal is not feasible for such networks, and often methods able to provide reasonably good solutions in manageable time are favored~\citep{leuven}.
However, some methods try to stay within reasonable time limits while delivering solutions close to optimal~\citep{Combo,GNNS,aref2023heuristic}.

Given the NP-hardness of CPP, exact solutions are rarely proposed.
Most of the existing approaches utilize the branch-and-bound method~\citep{Dorndorf1994,Jaehn2013CPP} or cutting plane technique~\citep{Grotschel1989cutting,Oosten2001}.
Many use both methods through the means of optimization software packages that internally implement them~\citep{Du2021SolvingCPPComparison}.
Usually, such works propose some extra steps to make the problem easier to solve with standard packages~\citep{Miyauchi2018exact,Lorena2019,Belyi2022SizeReduction}.
A few approaches were proposed for slight variations of CPP with different constraints, like branch-and-price for the capacitated or graph-connected version~\citep{mehrotra1998cliques,Benati2022BranchPrice} or branch-and-price-and-cut method for CPP with minimum clique size requirement~\citep{ji2007branchPriceCut}.
Integer programming models for clustering proved to be a useful tool~\citep{pirim2018JOTA}, and the vast majority of approaches try to solve CPP formulated as an integer linear programming~(ILP) problem~\citep{Grotschel1989cutting,Oosten2001,Miyauchi2018exact,Du2021SolvingCPPComparison}.
Researchers used similar methods in network science to maximize modularity~\citep{Agarwal2008,Aloise2010,dinh2015redundant,Lorena2019}.
In their seminal work, \cite{Agarwal2008} proposed solving the relaxation of~ILP to linear programming~(LP) problem and then rounding solution to integers.
They described a rounding algorithm that can provide a feasible solution to the initial problem, which (after applying local-search post-processing), in many cases, could achieve high modularity.
\added[id=A]{In the most recent work, \cite{aref2022bayan} proposed the Bayan algorithm grounded in an ILP formulation of the modularity maximization problem and relying on the branch-and-cut scheme for solving the problem to global optimality.}
While finding the \replaced[id=A]{global}{exact} maximum is unfeasible for large networks, studies in community detection showed that just providing the upper bound on achievable modularity could be useful by itself, and a few approaches were proposed recently~\citep{Miyauchi2013,Sobolevsky2017TrianglesArxiv}.

In this work, we present a new method for finding an upper bound on values that the objective function of CPP could reach.
By further developing the idea proposed by~\cite{Sobolevsky2017TrianglesArxiv}, we base our approach on combining known upper bounds of small subnetworks to calculate the upper bound for the whole network.
We describe how to use obtained upper bounds to construct the exact solution of CPP.
The proposed method is similar to the algorithm of \cite{Jaehn2013CPP} and its further development by~\cite{Belyi2019JBSU}.
However, by significantly improving the upper bound's initial estimates and recalculation procedure, it achieves a decrease of a couple of orders of magnitude in computational complexity and execution time.
Moreover, we show that our algorithm can find exact solutions to problems that the algorithm from~\cite{Jaehn2013CPP} could not.
In the end, we discuss possible directions of future research and show how adding new subnetworks could improve upper bound estimates.

\section{Problem Formulation and Existing Solution Approaches}\label{sect:formulation}

We consider the following problem.
Given a complete weighted undirected graph $G=(V, E, W)$, where $V=\{1\dots n\}$ is a set of vertices, $E=\{\{i, j\}\mid i,j\in V, i\ne j\}$ is a set of edges, $W=\{w_{ij} \in \mathbb{R} \mid \{i,j\} \in E, w_{ij}=w_{ji}\}$ is a set of weights of edges, find such a partition of its vertices $V$ into clusters (represented by a mapping function $C: V \rightarrow \mathbb{N}$ from vertices into cluster labels $c_v=C(v)$) that sum of edge weights within the clusters is maximized:
\begin{equation}
Q(G, C) = \sum_{1\le i<j\le n \mid c_i=c_j}{w_{ij}} \rightarrow \max \label{eq:Q}.
\end{equation}
We denote this sum as $Q$ and will refer to it as the partition quality function or CPP objective function.
Note that this problem can be defined for any graph by adding edges with zero weight, ignoring loop edges, and averaging the weights of incoming and outgoing edges.
We will say that in a given partition, an edge is \textit{included} (because its weight is included in the sum in equation~(\ref{eq:Q})) if it connects two nodes from the same cluster.
Otherwise, we will say that it is \textit{excluded}.
Also, we use the words graph/network and vertex/node interchangeably here and throughout the rest of the text.


\cite{Grotschel1989cutting} showed that CPP can be formulated as the following integer linear programming~\added[id=A]{(ILP)} problem.
For every edge $\{i, j\}$, we define a binary variable $x_{ij}$ that equals $1$ when the edge is included and $0$ otherwise.
Then the objective of CPP is to
\begin{equation}
\begin{split}
\text{maximize } & Q=\sum_{1\le i<j\le n}{w_{ij}\cdot x_{ij}},\\
\text{subject to } & x_{ij} + x_{jk} - x_{ik} \leq 1, \text{ for all } 1 \leq i < j < k \leq n \\
& x_{ij} - x_{jk} + x_{ik} \leq 1, \text{ for all } 1 \leq i < j < k \leq n \\
& -x_{ij} + x_{jk} + x_{ik} \leq 1, \text{ for all } 1 \leq i < j < k \leq n \\
& x_{ij} \in \{0,1\}, \text{ for all } 1 \leq i < j \leq n.
\end{split}
\label{eq:ILP}
\end{equation}
Constraints are called \textit{triangle inequalities} and ensure consistency of partition, i.e., if both edges $\{i, j\}$ and $\{j, k\}$ are included, then edge $\{i, k\}$ must be included too.

ILP formulation~(\ref{eq:ILP}) has been employed by many algorithms for CPP and its variants.
In their article, \cite{Grotschel1989cutting} empirically showed that many constraints are not saturated in the optimal solution and are redundant for the problem.
More recently, \cite{dinh2015redundant} derived a set of redundant constraints in formulation~(\ref{eq:ILP}) for modularity optimization, then \cite{Miyauchi2015Redundant} generalized Dinh and Thai’s results to the general case of CPP, and recently~\cite{koshimura2022concise} proposed even more concise formulation of ILP.
Developing their idea further, \cite{Miyauchi2018exact} proposed an exact algorithm that solves a modified ILP problem and then performs simple post-processing to produce an optimal solution to the original problem.

Extending the results of \cite{Grotschel1990facets}, \cite{Oosten2001} studied the polytope of~(\ref{eq:ILP}) and described new classes of facet-defining inequalities that could be used in a cutting plane algorithm.
\cite{sukegawa2013lagrangian} proposed a size reduction algorithm for (\ref{eq:ILP}) based on the Lagrangian relaxation and pegging test.
They showed that for some instances of CPP, their algorithm, which minimizes the duality gap, can find an exact solution.
For the other cases, they provided an upper bound of the solution.
Even without an exact solution, knowing an upper bound could be useful by itself~\citep{Miyauchi2013}.
For example, it allows estimating how good a particular solution found by a heuristic is.
The most common way to obtain upper bounds is to solve the problem~(\ref{eq:ILP}) with relaxed integrity constraints (i.e., when constraints $x_{ij}\in\{0,1\}$ are replaced with $x_{ij}\in[0,1]$). We refer to this version of the problem as the relaxed problem~(\ref{eq:ILP}).
In this case, the problem becomes an LP problem and can be solved in polynomial time using existing methods~\citep{Miyauchi2013}.

Looking at the problem from a different angle, \cite{Dorndorf1994} and \cite{Jaehn2013CPP} did not use formulation~(\ref{eq:ILP}).
Instead, they approached CPP as a combinatorial optimization problem and employed constraint programming to solve it.
In some sense, our approach combines both ideas:
we will show how to obtain tight upper bounds by solving another linear programming problem and then use the branch-and-bound method to solve CPP.

\section{Upper Bound Estimation}

In the general case of CPP, there is no theoretical limit on what values the quality function can reach since edge weights could be arbitrarily large.
For modularity scores, however, ~\cite{Brandes2008} proved that $-1/2\leq Q \leq 1$.
In practice, though, for every network $G=(V, E, W)$, a \textit{trivial upper bound} $\overline{Q}$ could be obtained simply as a sum of all positive edges:
\begin{equation}\label{eq:Q_trivial_max}
    \overline{Q}(G) = \sum_{\{i, j\}\in E \mid w_{ij}>0}{w_{ij}} \geq Q(G, C), \text{for any } C.
\end{equation}
But even this threshold is usually quite far above the actual maximum.
To further reduce this upper bound, \cite{Jaehn2013CPP} used triples of vertices in which two edges are positive (i.e., have positive weights) and one is negative (i.e., has negative weight).
However, their approach considered only edge-disjoint triples\added[id=A]{, i.e., triples of nodes that have no more than one common node and thus no shared edges}.
In what follows, we are developing a similar idea and generalizing it further.
We show how any subgraph, for which we know the upper bound of its partition quality function, could be used to reduce the upper bound for the whole network.
Furthermore, we also show how to account for overlaps of such subgraphs.
We start by introducing a few definitions.

\begin{figure}[h]
\centering
\includegraphics[width=1\textwidth]{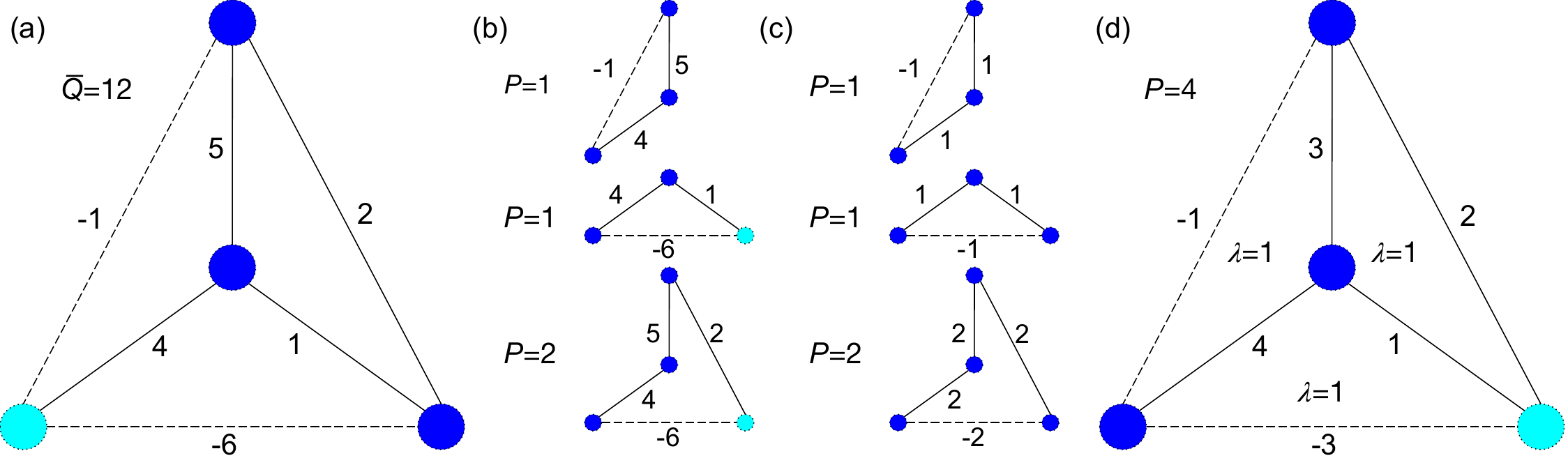}
\caption{\label{fig:Method}Illustration of definitions: a) Original network; b) Penalizing subnetworks (two triangles and a chain); c) Reduced subnetworks; d) Permissible linear combination of reduced subnetworks with all weights $\lambda$ equal to $1$. \added[id=A]{Colors indicate one of the possible optimal partitions. While subnetworks' optimal partitions do not determine the optimal partition of the network, their penalties can be used to estimate the network's penalty.}}
\end{figure}

\begin{definition}\label{Def:SN}
A \textit{subnetwork} $S=(V^*\subseteq V, E^*, W^*)$ is a complete network built on a subset of nodes of the original network $G=(V, E, W)$ defined in Section~\ref{sect:formulation}, where each edge $\{i, j\}\in E^*$ has weight $w^*_{ij}\in W^*$ such that $\lvert w^*_{ij} \rvert\leq \lvert w_{ij} \rvert$, $w_{ij}\in W$ and $w^*_{ij} \cdot w_{ij}\geq 0$, i.e., weights of subnetwork's edges have smaller or equal absolute values and the same sign (unless the weight is zero) as weights in the original network.
\end{definition}
For example, networks in Fig.~\ref{fig:Method}b and in Fig.~\ref{fig:Method}c are subnetworks of a network in Fig.~\ref{fig:Method}a.
For small networks with just a few nodes or with a simple structure, it is often easy to find the exact solution of CPP, e.g., by considering all possible partitions.
For some more complex networks, when finding the exact solution is already complicated, it might still be possible to prove tighter upper bound estimates $Q_{max}$ than the trivial one ($Q(C)\leq Q_{max} < \overline{Q}$ for any partition $C$). Our idea is to find such networks among subnetworks of the original graph $G$ and use their upper bound estimates to prove an estimate for~$G$.

\begin{definition}\label{Def:Penalty}
For subnetwork $S$ with an upper bound estimate $Q_{max}$, its \textit{penalty}~$(P)$ is the difference between the trivial upper bound of the objective function given by formula~(\ref{eq:Q_trivial_max}) and $Q_{max}$: $P=\overline{Q}(S)-Q_{max}$.
\end{definition}

We call a subnetwork with a positive penalty a \textit{penalizing subnetwork}.
For example, it is easy to see the best partitions of subnetworks in Fig.~\ref{fig:Method}b: we either keep all vertices in one cluster, including the negative edge or split them into two clusters, excluding the negative edge and the smallest positive edge.
This way, we know the optimal value of objective function $Q$, which we can use as a sharp upper bound.
Then the penalty is just the difference between the sum of positive edge weights and this value of $Q$.
We call any subnetwork of a given penalizing subnetwork $S$ having the same penalty as $S$ a \textit{reduced subnetwork}.
For example, subnetworks in Fig.~\ref{fig:Method}c are reduced subnetworks of their counterparts from Fig.~\ref{fig:Method}b.
\added[id=A]{The benefit of using them will become clearer by the end of this section.}

\begin{definition}\label{Def:LS}
Given a set of subnetworks $\{S_k=(V_k, E_k, W_k) \mid k=1\dots K\}$ of graph $G$, their \textit{permissible linear combination} $S^L=(V^L, E^L, W^L)$ with non-negative coefficients $\lambda_k$ is a subnetwork of the original network $G$ with all the same nodes $V^L=\bigcup_{k=1}^K{V_k}$ and edge weights equal to linear combinations $w^L_{ij}=\sum_{k=1}^K{\lambda_k w^{*k}_{ij}}$ of the corresponding weights 
$w^{*k}_{ij}=
\begin{cases}
  w^{k}_{ij} & \text{if } \{i, j\}\in E_k\\
  0 & \text{otherwise}
\end{cases}$.
\end{definition}

\added[id=A]{
The intuition here is that we re-weight and combine several subnetworks to get one.
With some abuse of notation, it can be written that $S^L=\sum_{k=1}^K \lambda_k S_k$, where multiplying a subnetwork by a scalar means multiplying all edge weights by this scalar, and summing subnetworks means uniting vertex sets and summing corresponding edge weights.
}
For example, the network in Fig.~\ref{fig:Method}d is a permissible linear combination of reduced subnetworks from Fig.~\ref{fig:Method}c.

Now it is easy to see that the following proposition holds:

\begin{lemma}[Summation lemma]\label{thrm:SummationLemma}
Consider a set of subnetworks of graph $G$ $\{S_1, S_2,..., S_K\}$ with penalties $P_1, P_2, ..., P_K$, and a permissible linear combination $S^L=\sum_{k=1}^K{\lambda_k S_k}$ with non-negative $\lambda_k$.
Then $S^L$ has a penalty greater or equal to $\sum_{k=1}^K{\lambda_k P_k}$.
\end{lemma}
\begin{proof}
Indeed, for each subnetwork $S_k$ denote the upper bound estimate corresponding to $P_k$ as $Q_{max}^k$.
Then for any partition $C$ of the subnetwork $S^L$, score $Q$ can be expressed as
$Q(S^L, C) = \sum_{1\le i<j\le n \mid \{i, j\}\in E^L, c_i=c_j}{w^L_{ij}} = \sum_{1\le i<j\le n \mid \{i, j\}\in E^L, c_i=c_j}{\sum_{k=1}^K{\lambda_k w^{*k}_{ij}}} = \sum_{k=1}^K{\lambda_k \sum_{1\le i<j\le n \mid \{i, j\}\in E_k, c_i=c_j}{w^{k}_{ij}}} = \sum_{k=1}^K{\lambda_k Q(S_k, C)} \leq \sum_{k=1}^K{\lambda_k Q_{max}^k}$,
so $Q^L_{max}=\sum_{k=1}^K{\lambda_k Q_{max}^k}$ is an upper bound for subnetwork $S^L$,
and since for the trivial upper bounds (\ref{eq:Q_trivial_max}) $\overline{Q}(S^L) = \sum_k{\lambda_k \overline{Q}(S_k)}$, $S^L$ has penalty $P^L = \overline{Q}(S^L)-Q^L_{max} = \sum_{k=1}^K{\lambda_k \overline{Q}(S_k)}-\sum_{k=1}^K{\lambda_k Q_{max}^k} = \sum_{k=1}^K{\lambda_k \left(\overline{Q}(S_k)-Q_{max}^k\right)} = \sum_{k=1}^K{\lambda_k P_k}$.
\end{proof}

Summation lemma allows us to prove a stronger result:

\begin{theorem}\label{thrm:PenaltyTheorem}
Consider graph $G$, a set of its subnetworks $\{S_1, S_2,..., S_K\}$ with penalties $P_1, P_2, ..., P_K$, and a permissible linear combination $S^L=\sum_{k=1}^K{\lambda_k S_k}$ with non-negative $\lambda_k$.
Then $G$ has a penalty greater or equal to $\sum_{k=1}^K{\lambda_k P_k}$.
\end{theorem}
\begin{proof}
Network $G$ can be represented as a sum $G = S^L + R$ of $S^L$ and some residual subnetwork $R$,
with edge weights $w^R_{ij} = w_{ij} - w^{*L}_{ij}$, where $w^{*L}_{ij}$ are equal to the edge weights $w^L_{ij}$ of $S^L$, if edge $\{i, j\}$ belongs to $S^L$, \replaced[id=A]{and}{or} zero otherwise.
Then,
$\overline{Q}(G) = \overline{Q}(S^L) + \overline{Q}(R)$,
and for any partition $C$,
$Q(G, C) = Q(S^L, C) + Q(R, C)$.
From the summation lemma, we have the following:
$Q(S^L, C) \leq \overline{Q}(S^L) - \sum_k{\lambda_k P_k}$.
So,
$Q(G, C) \leq Q(R, C) + \overline{Q}(S^L) - \sum_k{\lambda_k P_k} \leq
\overline{Q}(R) + \overline{Q}(S^L) - \sum_k{\lambda_k P_k} =
\overline{Q}(G) - \sum_k{\lambda_k P_k} = Q_{max}$,
and $G$ has penalty $P = \overline{Q}(G) - Q_{max} = \sum_k{\lambda_k P_k}$.
\end{proof}

This result provides a framework for constructing tight upper bounds by combining penalties of smaller subnetworks.
Having a set of penalizing subnetworks $\{S_1, S_2,..., S_K\}$ with their penalties $P_1, P_2, ..., P_K$, we can construct a linear programming problem to find the penalty of the whole network.
LP~problem can be formulated as follows:
\begin{equation}\label{eq:LPproblem}
\begin{split}
\text{maximize } & P=\sum_{k=1}^K{\lambda_k P_k},\\
\text{subject to } & \sum_{k=1}^K{\lambda_k \lvert w^{*k}_{ij}\rvert } \leq \lvert w_{ij} \rvert , \text{ for all } 1 \leq i < j \leq n,\\
& 0 \leq \lambda_k, \text{ for all } 1 \leq k \leq K.
\end{split}
\end{equation}
Constraints ensure that the linear combination of subnetworks \replaced[id=A]{remains a subnetwork, i.e., satisfies the condition $\lvert w^*_{ij} \rvert\leq \lvert w_{ij} \rvert$}{is permissible}.
\added[id=A]{Here comes the benefit of using \textit{reduced subnetworks}: having smaller edge weights $\lvert w^{*k}_{ij} \rvert$ while keeping the same penalty $P_k$ allows the possibility of finding larger coefficients~$\lambda_k$, and thus larger total penalty $P$.}
Note that this is not an integer problem and could be efficiently solved with modern optimisation software packages \added[id=A]{in polynomial time~\citep{Yin2015FasterLP,Cohen2021SolvingLP}}.
A large penalty $P$ found this way leads to a tight upper bound $Q_{max}=\overline{Q}-P$.

\begin{figure}[h]
\centering
\includegraphics[width=\textwidth]{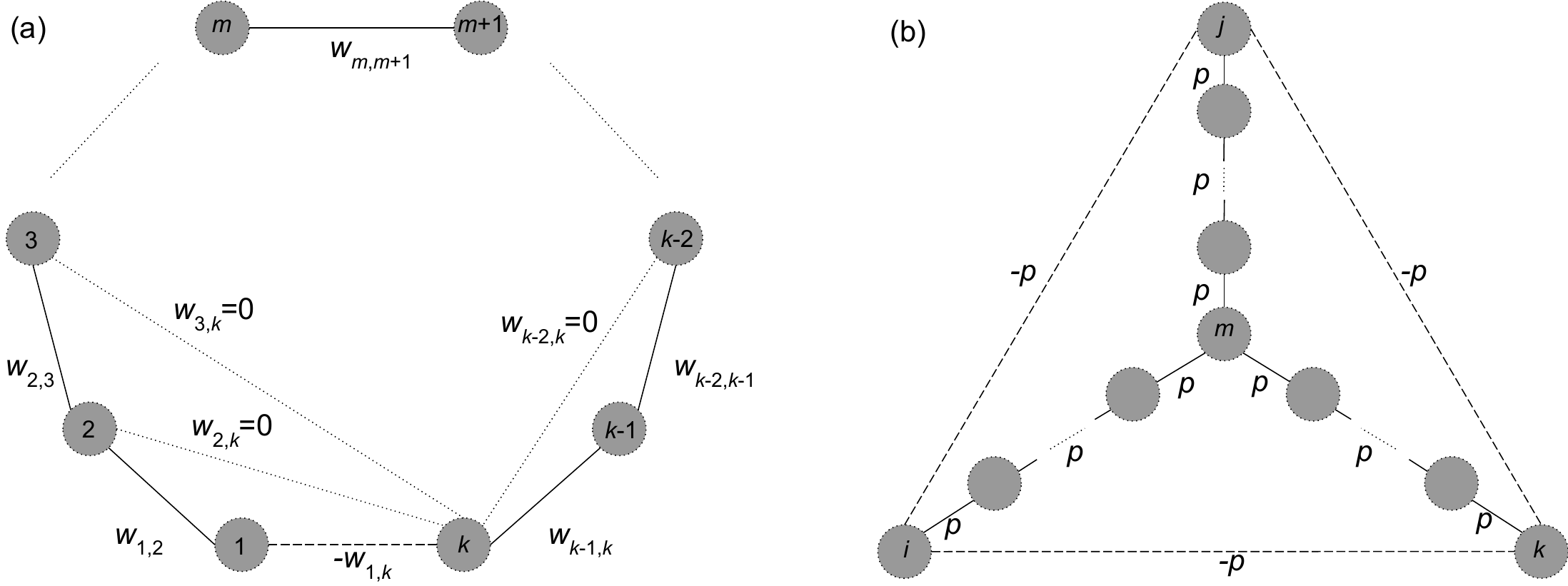}
\caption{\label{fig:ChainStar}Penalizing subnetworks: a) Chain; b) Star.}
\end{figure}

\subsection{Chains} 

Our method primarily focuses on a particular case of penalizing subnetworks that we call chains.
\begin{definition}\label{Def:Chain}
A \textit{chain} of length $k$ is a subnetwork consisting of $k$ nodes relabeled within a chain $1^*,\dots,k^*$, connected by positive edges $\{1^*, 2^*\}, \{2^*, 3^*\},\dots,\{(k-1)^*, k^*\}$ and a negative edge $\{1^*, k^*\}$.
When $k=3$, we call the chain a \textit{triangle}.
\end{definition}
Fig.~\ref{fig:ChainStar}a illustrates a chain in a general case, and in Fig.~\ref{fig:Method}b, we show triangles and a chain of length $4$.

It is always easy to find the exact solution of CPP for a chain.
Indeed, in the optimal partition, nodes $1^*$ and $k^*$ appear either in the same cluster or different ones.
In the first case, the negative edge $\{1^*, k^*\}$ is included in the total score, and in the second case, one of the positive edges must be excluded.
So, depending on which value is larger $\lvert w^*_{1^*,k^*} \rvert$ or $\min_{i=1,\dots,k-1}(w^*_{i^*,(i+1)^*})$, the optimal split of the chain will be into one or two clusters with the objective function value equal to $\sum_{i=1,\dots,k-1}{w^*_{i^*,(i+1)^*}} - \min(w^*_{1^*,2^*},\dots,w^*_{(k-1)^*,k^*},\lvert w^*_{k^*,1^*} \rvert)$,
where $\min(w^*_{1^*,2^*},\dots,w^*_{(k-1)^*,k^*},\lvert w^*_{k^*,1^*} \rvert)=P$ is that chain's penalty.
Then to construct a reduced chain, we can set the weight of each positive edge to $P$ and the weight of the negative edge to $-P$.
Repeating the same reasoning, one can easily see that this chain has the same penalty $P$.
In some sense, this is the best reduction possible because assigning a smaller absolute value to any weight in the original chain would lead to a smaller penalty.

Using chains alone, we can already calculate a non-trivial upper bound: find as many chains as possible, reduce them, construct an LP problem, and solve it using an appropriate method to obtain penalty $P$.
Then the upper bound is the difference $Q_{max} = \overline{Q} - P$.
By as many as possible, we mean as many as we can find and a solver can handle in a reasonable time (in our experiments, we used all penalizing chains with three and four nodes).
The following algorithm formalizes these steps.

\begin{algorithm2e}[H]
\caption{Calculate penalty by solving LP problem}\label{Alg:CalcPenaltyLP}
\SetKwInOut{Input}{input}
\SetKwInOut{Output}{output}
\Input{Graph $G$ represented as weight matrix $W=(w_{ij})$}
\Output{Penalty $P$ and a set of penalizing chains $Ch$}
\BlankLine
{find all chains of length 3 and 4}\tcp*[l]{Using four nested loops}
{construct LP problem}  \tcp*[l]{Using equations (\ref{eq:LPproblem})}
{solve LP problem, obtaining total penalty $P$ and weights of the chains}\;
{$Ch$ = chains with positive weights}\;
\Return{$P$, $Ch$\;}
\end{algorithm2e}

\begin{figure}[h]
\centering
\includegraphics[width=1\textwidth]{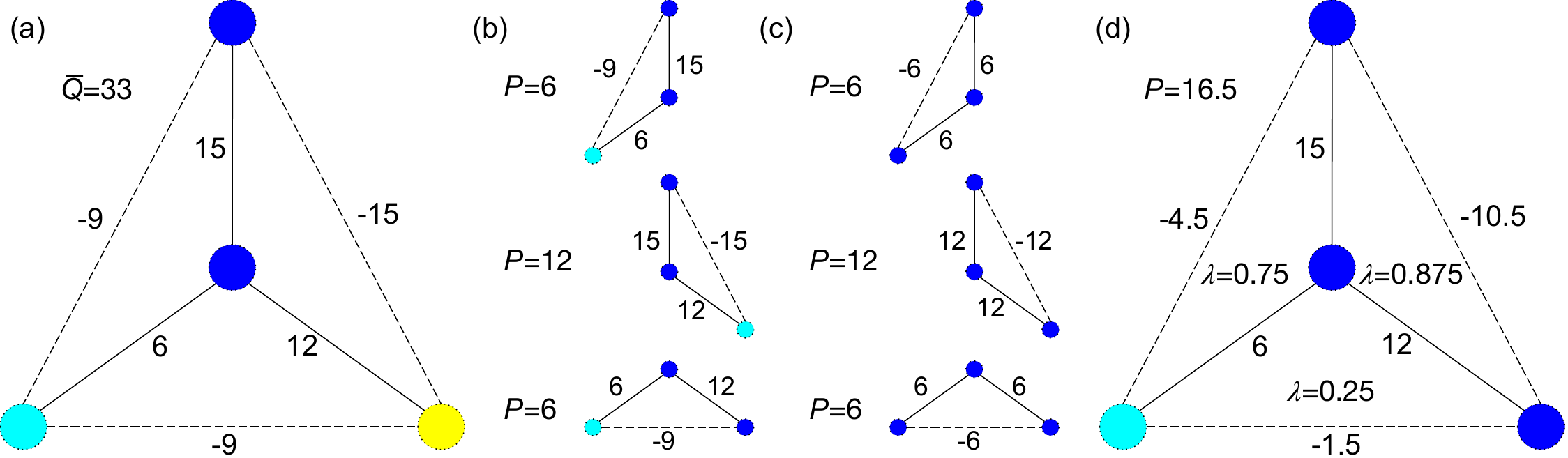}
\caption{\label{fig:Unresolvable}An example of a network for which the upper bound constructed using chains does not match the optimal objective function value. Original network (a) has only three chains (b). After reducing them (c) and constructing optimal permissible linear combination (d) by solving an LP problem, the best penalty is $16.5$, and the upper bound is $33-16.5=16.5$. However, it is easy to see (by considering all possibilities) that the best partition of the original network has a~score~of~$15$ and a~penalty~of~$18$.
\added[id=A]{Colors indicate one of the possible optimal partitions.}}
\end{figure}

We note here that while adding more chains, in general, helps to construct tighter upper bounds, it is possible that even after considering all chains, the upper bound will not become sharp, i.e., it will not reach the optimal value of the objective function.
We show an example of such a network in Fig.~\ref{fig:Unresolvable}.
Moreover, we proved that solution to relaxed problem~(\ref{eq:ILP}) (i.e., when constraints $x_{ij}\in\{0,1\}$ are replaced with $x_{ij}\in[0,1]$) always finds at least as tight upper bounds as considering only chains as penalizing subnetworks (see Appendix~\ref{sect:Appendix}).
Nevertheless, we will show that chains alone can already give a tight upper bound for many networks while the resulting LP problem is smaller and faster to solve.
Furthermore, the main advantage of our framework is that it allows us to use not only chains but any penalizing subnetworks.

\subsection{Stars} 

To show an example of penalizing subnetworks that can help to construct upper bounds tighter than those found by solving the relaxed problem~(\ref{eq:ILP}), we introduce one more class of penalizing subnetworks that we call \textit{stars}.
The intuition comes from the example in Fig.~\ref{fig:Unresolvable}.
A star is a network that has three nodes $i, j, k$ connected to each other by edges with weight $-p$, and a node $m$ connected to $i, j, k$ via simple non-overlapping paths consisting of edges with weight $p$, for some positive number $p$ (see Fig.~\ref{fig:ChainStar}b).
The solution of relaxed problem~(\ref{eq:ILP}) for a star gives a penalty of $1.5p$.
So it follows that the actual penalty is at least $2p$, and it is easy to see how it can be achieved.

\section{Branch and Bound\label{sect:b&b}}

Here we describe how to use the method proposed in the previous section inside a general branch-and-bound technique to solve CPP.
In each step of branch and bound, we select an edge of the network and \textit{fix} it, i.e., consider two possibilities:
1) This edge is \textit{included}, i.e., lays within some cluster, so its weight is included in the total sum of the objective function.
So, the two nodes that it connects belong to the same cluster in the final partition.
2) This edge is \textit{excluded}, i.e., it connects nodes belonging to different clusters in the final partition, and so its weight is not included in the objective function score.
In each of these two cases, we recalculate the estimate of what the objective function score could be, and if the upper bound is equal to or smaller than the value achieved by some already known feasible solution, then the case cannot lead to a better solution, so it is fathomed.
Then, for each case that is not fathomed, the same steps are repeated recursively.
This procedure creates a binary search tree that is being traversed depth-first.

There are a few things we need to consider at each step.
First, we must ensure that constraints imposed by edge inclusion or exclusion are not contradictory.
That means we need to propagate transitivity condition: if edges between $a$ and $b$, and $b$ and $c$ are included, then the edge between $a$ and $c$ must be included too; and if edge $\{a,b\}$ is included and edge $\{b,c\}$ is excluded, then edge $\{a,c\}$ must be excluded.
To ensure this, we use the following algorithm.

\begin{algorithm2e}[H]
\caption{Transitivity constraints propagation}\label{Alg:Transitivity}
\SetKwInOut{Input}{input}
\SetKwInOut{Output}{output}
\Input{Graph $G$, set of already fixed edges for which transitivity is already satisfied, and a newly fixed edge $\{a,b\}$}
\Output{Updated set of fixed edges in $G$, where transitivity condition is satisfied again}
\BlankLine
{define four initially empty sets of vertices $A,B,X,Y$}\;
\ForEach{vertex $u$ in graph $G$}{
  \lIf{edge $\{u,a\}$ is included}{add $u$ to $A$}
  \lElseIf{edge $\{u,a\}$ is excluded}{add $u$ to $X$}
  \lIf{edge $\{u,b\}$ is included}{add $u$ to $B$}
  \lElseIf{edge $\{u,b\}$ is excluded}{add $u$ to $Y$}
}
\eIf{edge $\{a,b\}$ is included}{
  include all edges between vertices of $A$ and $B$\;
  exclude all edges between vertices of $A$ and $Y$\;
  exclude all edges between vertices of $B$ and $X$\;
} {
  exclude all edges between vertices of $A$ and $B$\;
}
\Return{updated set of fixed edges\;}
\end{algorithm2e}

The correctness of this algorithm follows from the observation that when the transitivity condition is satisfied, nodes connected by included edges form cliques, and every node is connected to all nodes in such a clique via the same type of edges (included, excluded, or non-fixed).
Fixing edge $\{a, b\}$ may break this property, and Algorithm~\ref{Alg:Transitivity} restores it.
The time complexity of this algorithm is $O(n^2)$, where $n=\lvert V \rvert$ is the number of nodes in $G$ because the number of all possible edges between sets $A, B, X, Y$ is not greater than the total number of edges, which is $n(n-1)/2$.
However, in practice, in many cases, we do not need to fix many edges, so sets $A, B, X, Y$ are small enough so that their sizes could be considered constant, as well as the number of edges between them, then the complexity of the algorithm is dominated by the first loop over all nodes, and thus this algorithm runs in linear time $O(n)$.

A second consideration is that we want to update our upper bound estimate after each edge fixation.
We do so by noting that if we include a negative or exclude a positive edge $\{i, j\}$, then $\lvert w_{ij} \rvert$ should be added to the network's penalty since $Q\leq\overline{Q}-\lvert w_{ij} \rvert$ in this case.
However, any fixation of edge $\{i, j\}$ changes the penalties of some subnetworks in which it is present, so the penalty of each affected subnetwork needs to be recalculated.
Fortunately, it is easy to do for chains and stars.
If we include a negative edge or exclude a positive one, then adding $\lvert w_{ij} \rvert$ to the network's total penalty entirely accounts for any penalty incurred by this edge in any chain containing it, so we should stop considering such chains.
We do the same for stars, but the reason is a bit less apparent.
We stop considering stars containing edge $\{i, j\}$ because their penalties are accounted for by weight $\lvert w_{ij} \rvert$ added to the total penalty and by penalties of chains that do not go through $\{i, j\}$.
In contrast, if we exclude a negative or include a positive edge $\{i, j\}$, the total penalty is not affected directly.
To account for such fixation, we need to exclude this edge from the constraints in LP formulation~(\ref{eq:LPproblem}).

In our experiment, we implemented the branch-and-bound algorithm using only chains to calculate upper bounds.
Separately, we implemented a version where we included stars to estimate the initial upper bound.
In the rest of this section, we will present algorithms for chains only. They could be easily generalized to use stars too.
However, in our experiments, the benefit of tighter upper bounds obtained from solving the LP problem with stars was offset by the larger LP problem that was slower to solve.

\begin{algorithm2e}[H]
\caption{Calculate penalty using heuristic}\label{Alg:CalcPenaltyHeuristic}
\SetKwInOut{Input}{input}
\SetKwInOut{Output}{output}
\SetKwFunction{RandomShuffle}{RandomShuffle}
\SetKwFunction{FindShortestPositivePath}{FindShortestPositivePath}
\SetKwRepeat{Do}{do}{while}
\SetKw{And}{and}
\SetKw{Or}{or}
\SetKw{True}{true}
\SetKw{False}{false}
\SetKw{Break}{break}
\Input{Graph $G$, set of fixed edges $F$, previous set of chains $Ch$}
\Output{Penalty $P$ and a set of penalizing chains $Ch$}
\BlankLine
{define penalty $P=0$}\;
{define a new set of chains $Ch_{new}=\varnothing$\;}
\ForEach{chain $c$ with penalty $p$ in $Ch$}{
    {define boolean flag $keep\_chain$ = \True}\;
    \ForEach{edge $(u,v)$ in $c$}{
      \lIf{$(u,v)$ is included \And $w^*_{uv}<0$ \Or
             $(u,v)$ is excluded \And $w^*_{uv}>0$}{
          $keep\_chain$ = \False 
      }
    }
  \If{$keep\_chain$}{
    $Ch_{new} = Ch_{new} \cup \{c\}$ \tcp*[l]{Add $c$ to new set of chains}
    $P = P + p$\;
    \ForEach{not fixed edge $(u,v)$ in $c$}{
      \lIf{$w_{uv}>0$}{ $w_{uv} = w_{uv} - p$ }
      \lElse{ $w_{uv} = w_{uv} + p$ }
    }
  }
}
\ForEach{$len$ = 2 to Infinity}{
  \lIf{there are no negative edges in positive connected components in $G$}{\Break}
  find all negative edges $E_{neg}$ in $G$\;
  \ForEach{edge $(u,v)$ in \RandomShuffle{$E_{neg}$}}{
    \While{$w_{uv}<0$}{
      $path$ = \FindShortestPositivePath{$G$, $u$, $v$} \tcp*[l]{simple BFS}
      \lIf{path length $>len$}{\Break}
      {construct chain $c$ from $path$ and edge $(u,v)$}\;
      {calculate penalty $p$ of chain $c$}\;
      \lForEach{not fixed edge $(i,j)$ in $path$}{
        $w_{ij} = w_{ij} - p$
      }
      $w_{uv} = w_{uv} + p$\;
      $P = P + p$\;
      $Ch_{new} = Ch_{new} \cup \{c\}$ \tcp*[l]{Add $c$ to new set of chains}
    }
  }
}
\Return{$P$, $Ch_{new}$\;}
\end{algorithm2e}

Moreover, it appeared that instead of solving the LP problem at each step, it is often expedient to use a much faster greedy technique that produces less tight upper bounds.
The idea is to use a good set of chains with their weights $\lambda$ already found at previous steps, and instead of constructing and solving the LP problem for the whole network, construct a residual subnetwork $R = G - S^L$ and find new chains in it using a simple heuristic: find a (random) chain and subtract it from $R$ (with weight $\lambda=1$), then find another (random) chain and subtract it, repeat this process, until there are no more chains. This process is formalized in algorithm~\ref{Alg:CalcPenaltyHeuristic}.

This method works quickly, but it accumulates inefficiencies.
To deal with them, after considering some number of levels (e.g. four) in the branch-and-bound search tree, we still solve the complete LP problem to update chains and their weights.
For that purpose, we can use a slightly adjusted algorithm~\ref{Alg:CalcPenaltyLP} that takes into account fixed edges.

Now we are ready to present the main workhorse of branch and bound: a recursive function that explores each node of the tree, i.e., tries to include and exclude an edge and calls itself recursively to explore the search tree further.

\begin{algorithm2e}[H]
\caption{Recursive branching}\label{Alg:DFS}
\SetKwInOut{Input}{input}
\SetKwInOut{Output}{output}
\SetKwFunction{UpdateFixedEdges}{UpdateFixedEdges}
\SetKwFunction{CalcPenaltyHeuristic}{CalcPenaltyHeuristic}
\SetKwFunction{CalcPenaltyLP}{CalcPenaltyLP}
\SetKwFunction{DFS}{RecursiveBranching}
\SetKw{And}{and}
\SetKw{Or}{or}
\Input{Graph $G$, list of positive edges $L$, current edge index $e$, set of fixed edges $F$, the best partition $C$ found so far and its score $Q_{min}$, current recursion depth $d$}
\Output{New the best feasible partition and its quality score}
\BlankLine
\lWhile{$L[e]$ is fixed}{$e = e + 1$}
\If(\tcp*[h]{we have fixed all positive edges}){$e > |L|$}{
  exclude all negative edges that are not fixed yet\;
  \Return{current partition, current $Q$\;}
}
\ForEach{\{include edge $L[e]$, exclude edge $L[e]$\}}{
  $F'$ = \UpdateFixedEdges{$G$,$F$, $L[e]$}  \tcp*[l]{Algorithm~\ref{Alg:Transitivity}}
  {define penalty $P_0=0$}\;
  \ForEach{edge $(u,v)$ in set $F'$}{
    \lIf{$(u,v)$ is included \And $w_{uv}<0$ \Or
        $(u,v)$ is excluded \And $w_{uv}>0$}{
      $P_0 = P_0 + \lvert w_{uv} \rvert$
    }
  }
  \tcp{every few steps we try to obtain a higher penalty}
  \eIf{$d \mod 4 == 0$}{
    $P$, $Ch$ = \CalcPenaltyLP{$G$, $F'$} \tcp*[l]{Adjusted algorithm~\ref{Alg:CalcPenaltyLP}}
  }{
    $P$, $Ch$ = \CalcPenaltyHeuristic{$G$, $F'$, $Ch$} \tcp*[l]{Algorithm~\ref{Alg:CalcPenaltyHeuristic}}
  }
  \If(\tcp*[h]{Recursive call}){$\overline{Q} - P_0 - P > Q_{min}$}{
    $C$, $Q_{min}$ = \DFS{$G$, $L$, $e+1$, $F'$, $C$, $Q_{min}$, $d+1$}\;
  }
}
\Return{$C$, $Q_{min}$\;}
\end{algorithm2e}

Two more notes before we can finally formulate the main procedure:
1)~The order in which we consider edges influences performance quite a lot.
However, calculating penalty change after fixing each edge on every step is too computationally expensive.
So, we used edge weights as an approximation.
The reasoning here is that excluding a heavy positive edge would cause a higher loss in partition score.
%
2)~For some most simple cases, even the heuristic of algorithm~\ref{Alg:CalcPenaltyHeuristic} can already find an upper bound that matches a feasible solution proving its optimality.
So to avoid spending time on solving the LP problem, we try a few (e.g. three) times to calculate the initial upper bound using the heuristic.

The following algorithm describes the main procedure of the branch-and-bound method that calls the functions presented above to find a CPP solution.

\begin{algorithm2e}[H]
\caption{Branch and bound}\label{Alg:BnB}
\SetKwInOut{Input}{input}
\SetKwInOut{Output}{output}
\SetKwFunction{DFS}{RecursiveBranching}
\SetKwFunction{SortEdges}{SortEdges}
\SetKwFunction{CalcPenaltyHeuristic}{CalcPenaltyHeuristic}
\SetKwFunction{CalcPeanaltyLP}{CalcPeanaltyLP}
\SetKwFunction{GetFeasibleSolution}{GetFeasibleSolution}
\SetKwFor{repeat}{repeat}{times}
\Input{Graph $G$}
\Output{Optimal partition $C$ and its quality score $Q_{opt}$}
\BlankLine
{$C$, $Q_{min}$ = \GetFeasibleSolution{G}} \tcp*[l]{Run heuristic}
\repeat{3}{
    {$P$, $Ch$ = \CalcPenaltyHeuristic{G,$\varnothing$,$\varnothing$}} \tcp*[l]{Algorithm~\ref{Alg:CalcPenaltyHeuristic}}
    \lIf{$\overline{Q} - P == Q_{min}$}{
        \Return{$C$, $Q_{min}$}
    }
}
{$P$, $Ch$ = \CalcPeanaltyLP{G}} \tcp*[l]{Algorithm~\ref{Alg:CalcPenaltyLP}}
\lIf{$\overline{Q} - P == Q_{min}$}{
  \Return{$C$, $Q_{min}$}
}
{$L$ = positive edges of $G$ sorted in decreasing order of weight}\;
{$C$, $Q_{opt}$ = \DFS{$G$,$L$,1,$\varnothing$,$C$,$Q_{min}$,1}} \tcp*[l]{Algorithm~\ref{Alg:DFS}}
\Return{$C$, $Q_{opt}$\;}
\end{algorithm2e}

Since both our method and the method by~\cite{Jaehn2013CPP} implement the standard branch-and-bound technique, they both have a similar recursive structure with the same steps.
However, each step of the two methods is implemented differently:
1)~We use Combo~\citep{Combo} to obtain lower bounds, while Jaehn and Pesch use a heuristic by~\cite{Dorndorf1994}.
Our experiments show that Combo is more efficient\added[id=A]{, agreeing with recent results~\citep{aref2023heuristic}.}
2)~We use a more efficient algorithm~\ref{Alg:Transitivity} for constraints propagation.
3)~The order in which we consider edges is different.
4)~And, most importantly, we use methods introduced above to calculate much tighter upper bounds using penalizing subnetworks.

\section{Computational Experiment\label{sect:Results}}

We considered both proprietary solvers and open-source solutions to solve the LP problem that arises when calculating the upper bound estimate.
Based on the results of the comparison of some of the most popular open-source solvers~\citep{SolversComparison}, we picked COIN\_OR CLP~\citep{CLP} as an open-source candidate for our experiments.
Then, after it showed results similar to and often even better than proprietary solutions, we stayed with it as our linear programming solver.
Another argument in favour of an open-source solution was that we wanted our results to be freely available to everyone.
Our code and the datasets generated and analysed in this section are available on GitHub\footnote{\url{https://github.com/Alexander-Belyi/best-partition}}.
To find an initial solution of CPP, we used the algorithm Combo~\citep{Combo}, whose source code is also freely available.
All programs were implemented in C++, compiled using Clang 13.1.6, and ran on a laptop with a 3.2 GHz CPU and 16 GB RAM.

Among the recent works, there are two methods for solving CPP exactly that show the best results.
Our method extends and improves upon the one by \cite{Jaehn2013CPP}.
The other one is by \cite{Miyauchi2018exact}, whose main idea was to provide a way of significantly reducing the number of constraints in problem~(\ref{eq:ILP}) for networks with many zero-weight edges.
So they tested their algorithm only on networks with many zero-weight edges, which is the hard case for our method because a chain cannot be penalizing if it has non-fixed edges with weight zero.
Thus, the comparison with their results would be unfair.
Therefore we mostly adopted the testing strategy of \cite{Jaehn2013CPP} and compared results with their method, which is much more similar in spirit to ours.

\added[id=A]{We note that \cite{Jaehn2013CPP} did not make their implementation available, and our attempts to re-implement their method did not show any improvements compared to already reported results, sometimes falling behind significantly. Therefore, below, we compare our results with the results from their original paper. Although we use a modern laptop, running our algorithm on a computer from 2011 with similar characteristics to the one used by \cite{Jaehn2013CPP} decreases the performance only by a factor from 1.5 to 3, which is expected for a simple single-threaded program without heavy memory usage, like our algorithm. So, we believe the difference in laptop configuration cannot lead to misjudgments in our comparison.}

Also, following \cite{Jaehn2013CPP}, we report times for solving ILP problem~\ref{eq:ILP} using CPLEX Optimization Studio 20.1.
We use the version of problem~\ref{eq:ILP} without redundant constraints as proposed by~\cite{koshimura2022concise}.
\added[id=A]{We tried incorporating custom propagation techniques into CPLEX but could not achieve performance gains compared to the default settings.}
Finally, we note that there exists a benchmark for evaluating heuristic approaches adopted, for example, by \cite{hu2020two-model} and \cite{lu2021hybrid}, but it consists of instances too large to be solved exactly, and therefore we did not use it.

We tested our approach on both real-world and artificial networks.
The real-world networks were collected from previous studies found in the literature, and artificial networks are random graphs generated according to specified rules.
\cite{Jaehn2013CPP} considered two sets of real-world networks.
The first set studied by \cite{Grotschel1989cutting} was obtained by reducing an object clustering problem to CPP.
Some of the networks were published in the article's appendix, but some were only referenced, so we could not find Companies network.
For the network UNO, we got the same results as \cite{Grotschel1989cutting}, and they are slightly different from the results of~\cite{Jaehn2013CPP}, probably because of the typo in the data.
\cite{Jaehn2013CPP} mentioned this issue.
We present results for these networks in Table~\ref{tab:cpp_real_world_1}.
In all the tables that follow:
in column Nodes, we show the number of nodes considered by the branch-and-bound technique;
$t$ indicates execution time in seconds;
$n$ is the network's size;
$\overline{Q}$ is a trivial upper bound estimate~(\ref{eq:Q_trivial_max});
$Q_{min}$ represents the initial value obtained by a heuristic (Combo in our case);
$Q_{max}$ is the upper bound obtained on the first step by using penalizing chains in our algorithm and using triangles in the method of \cite{Jaehn2013CPP};
$Q_{opt}$ is the optimal solution;
asterisk~($^*$) indicates the results of the algorithm proposed here;
$t^{CPLEX}$ is execution time for CPLEX;
results of \cite{Jaehn2013CPP} \added[id=A]{($\overline{Q}$, $Q_{min}$, $Q_{max}$, Nodes, $t$)} are taken from their article;
better values are shown in bold.

\footnotesize
\begin{longtable}{| l | r | r | r | r | r | r | r | r | r | r |}
\hline
Network & n & $\overline{Q}$ & $Q_{max}^{*}$ & $Q_{max}$ & $Q_{opt}$ & Nodes$^{*}$ & Nodes & $t^*$\,(s) & $t$\,(s) & $t^{\tiny{CPLEX}}$\,(s) \\
\hline
\endfirsthead
\hline
\endhead
\hline
\endfoot

Wild cats &  30 &  1\,400 & \textbf{ 1\,304} &      1\,328  &  1\,304 & \textbf{ 0} &      92  & \textbf{0.00} &         0.03  &  0.16 \\
Cars      &  33 &  1\,748 & \textbf{ 1\,501} &      1\,589  &  1\,501 & \textbf{ 0} &     425  &         0.15  & \textbf{0.08} &  0.22 \\
Workers   &  34 &  1  233 & \textbf{    964} &      1\,020  &     964 & \textbf{29} &  2\,028  & \textbf{0.23} &         0.32  &  0.24 \\
Cetacea   &  36 &     998 & \textbf{    967} &         969  &     967 & \textbf{ 0} &       1  & \textbf{0.00} &         0.01  &  0.14 \\
Micro     &  40 &  1\,362 & \textbf{ 1\,034} &      1\,116  &  1\,034 & \textbf{ 0} & 21\,101  & \textbf{0.12} &         3.16  &  0.32 \\
UNO       &  54 &     918 & \textit{    798} & \textit{785} &     798 & \textbf{ 0} &      61  &         0.24  & \textbf{0.04} &  0.73 \\
UNO 1a    & 158 & 12\,322 &         12\,197  &     12\,197  & 12\,197 &          0  &       0  & \textbf{0.05} &         0.08  & 41.71 \\
UNO 1b    & 139 & 11\,859 &         11\,775  &     11\,775  & 11\,775 &          0  &       0  & \textbf{0.03} &         0.06  & 29.46 \\
UNO 2a    & 158 & 73\,178 & \textbf{72\,820} &     72\,874  & 72\,820 & \textbf{ 0} &     133  & \textbf{0.06} &         0.20  & 35.81 \\
UNO 2b    & 145 & 72\,111 & \textbf{71\,818} &     71\,840  & 71\,818 & \textbf{ 0} &     141  & \textbf{0.03} &         0.16  & 27.74 \\

\hline
\caption{Results of evaluation on the first set of real-world networks compiled by \cite{Grotschel1989cutting} compared with results of \cite{Jaehn2013CPP}.
\label{tab:cpp_real_world_1}}
\end{longtable}
\normalsize

It could be seen that all instances were solved by our method within a second.
Combo had already found the optimal solution in all cases, and our method was applied only to prove its optimality.
There are a few cases where the method of \cite{Jaehn2013CPP} was faster due to the quick heuristic they use to construct upper bounds, while our approach had to solve the LP problem.
However, we would notice that we did not have to use branch and bound for any network except Workers because the constructed upper bound was already equal to the lower bound found by Combo.

The second set of real-world networks arises from a part-to-machine assignment problem, which is often encountered in group technology, and was studied by~\cite{Oosten2001}.
Unfortunately, they did not publish their networks and only provided citations to sources.
Nevertheless, we obtained five out of the seven networks they considered.
These networks are particularly hard to solve because they are bipartite, which means that every triple of nodes has an edge with zero weight, so there are no triangles.
Just as~\cite{Oosten2001} and~\cite{Jaehn2013CPP}, we quickly solved three easy problems, but unlike them, we also solved MCC and BOC problems, although, CPLEX showed even better time.
Summary statistics are present in Table~\ref{tab:cpp_real_world_2}, but neither~\cite{Oosten2001} nor~\cite{Jaehn2013CPP} reported their execution time.

\footnotesize
\begin{longtable}{| l | r | r | r | r | r | r | r |}
\hline
Network & n & $\overline{Q}$ & $Q_{max}^{*}$ & $Q_{opt}$ & Nodes$^{*}$ & $t^{*}$\,(s) & $t^{CPLEX}$\,(s) \\
\hline
\endfirsthead
\hline
\endhead
\hline
\endfoot

KKV &  24 &   32 & 23.0  & 23 &       24 & \textbf{0.02} &   0.08 \\
SUL &  31 &   71 & 48.0  & 46 &        8 & \textbf{0.05} &   0.87 \\
SEI &  33 &   77 & 55.7  & 54 &       34 & \textbf{0.11} &   0.37 \\
MCC &  40 &   85 & 56.7  & 43 &  16\,095 &        95.35  & \textbf{22.03} \\
BOC &  59 &  126 & 84.0  & 67 & 106\,620 &    1\,494.21  & \textbf{156.1}  \\
\hline
\caption{Results of evaluation on the second set of real-world networks compiled by \cite{Oosten2001}. \label{tab:cpp_real_world_2}}
\end{longtable}
\normalsize

To generate random graphs, we repeated the procedures described by~\cite{Jaehn2013CPP}.
Similarly, we created four sets of synthetic networks.
The first set consists of graphs with $n$ vertices where $n$ ranges from $10$ to \replaced[id=A]{$23$}{$20$}\added[id=A]{. In the original paper by~\cite{Jaehn2013CPP}, authors used only networks of sizes up to $20$ nodes, but we extended all four datasets with networks of $21-23$ nodes for better comparison with CPLEX}.
In the first dataset, edge weights were selected uniformly from the range $[-q, q]$.
For every $n$ and every $q$ from a set $\{1, 2, 3, 5, 10, 50, 100\}$ we generated five random graphs, resulting in $35$ graphs for each $n$ or \replaced[id=A]{$490$}{$385$} graphs in~total.
Results for this set are shown in Table~\ref{tab:cpp_rand_1}.
Each value corresponds to a sum over $35$ instances.
Because our networks are different from those generated by~\cite{Jaehn2013CPP},
first, after each experiment, we divided $\overline{Q}$, $Q_{min}$, and $Q_{max}$ by $Q_{opt}$ and operated with relative numbers instead of absolute values of the objective function.
Second, we ran every experiment ten times with different random instances and reported the mean and the unbiased standard deviation estimate (as $mean \pm std.$).

\footnotesize
\begin{longtable}{| r | c | c | c | c | c | c |}
\hline
n & $\overline{Q}^{*}$ & $\overline{Q}$ & $Q_{min}^{*}$ & $Q_{min}$ & $Q_{max}^{*}$ & $Q_{max}$\\
\hline
\endfirsthead
\hline
\endhead
\hline
\endfoot
10  & 1.749 $\pm$ 0.048 & 1.764 & \textbf{0.998} $\pm$ 0.003 &         0.994  & \textbf{1.014} $\pm$ 0.005 & 1.226 \\
11  & 1.807 $\pm$ 0.046 & 1.831 & \textbf{0.995} $\pm$ 0.005 &         0.988  & \textbf{1.018} $\pm$ 0.007 & 1.272 \\
12  & 1.844 $\pm$ 0.053 & 1.932 & \textbf{0.998} $\pm$ 0.002 &         0.993  & \textbf{1.020} $\pm$ 0.007 & 1.305 \\
13  & 1.934 $\pm$ 0.049 & 1.867 & \textbf{0.997} $\pm$ 0.003 &         0.986  & \textbf{1.032} $\pm$ 0.010 & 1.287 \\
14  & 2.015 $\pm$ 0.044 & 1.971 & \textbf{0.996} $\pm$ 0.003 &         0.983  & \textbf{1.049} $\pm$ 0.015 & 1.355 \\
15  & 2.046 $\pm$ 0.029 & 2.071 & \textbf{0.997} $\pm$ 0.002 &         0.996  & \textbf{1.056} $\pm$ 0.013 & 1.367 \\
16  & 2.088 $\pm$ 0.036 & 2.043 &         0.996  $\pm$ 0.002 & \textbf{0.999} & \textbf{1.068} $\pm$ 0.018 & 1.341 \\
17  & 2.152 $\pm$ 0.053 & 2.189 &         0.995  $\pm$ 0.003 & \textbf{0.997} & \textbf{1.090} $\pm$ 0.020 & 1.419 \\
18  & 2.205 $\pm$ 0.045 & 2.230 & \textbf{0.994} $\pm$ 0.004 &         0.993  & \textbf{1.109} $\pm$ 0.022 & 1.433 \\
19  & 2.236 $\pm$ 0.036 & 2.236 & \textbf{0.993} $\pm$ 0.006 &         0.994  & \textbf{1.123} $\pm$ 0.017 & 1.439 \\
20  & 2.313 $\pm$ 0.037 & 2.251 & \textbf{0.993} $\pm$ 0.005 &         0.988  & \textbf{1.159} $\pm$ 0.018 & 1.440 \\
21  & 2.327 $\pm$ 0.048 &       &         0.994  $\pm$ 0.003 &                &         1.165  $\pm$ 0.023 &       \\
22  & 2.399 $\pm$ 0.041 &       &         0.993  $\pm$ 0.002 &                &         1.200  $\pm$ 0.021 &       \\
23  & 2.417 $\pm$ 0.047 &       &         0.993  $\pm$ 0.004 &                &         1.209  $\pm$ 0.023 &       \\
\hline
n & Nodes$^{*}$ & Nodes & $t^*$\,(s) & $t$\,(s) & \multicolumn{2}{c|}{$t^{CPLEX}$\,(s)} \\
\hline
10  & \textbf{     133.8} $\pm$ 48.30      &        1\,205 & \textbf{ 0.03} $\pm$  0.01  &     0.05 & \multicolumn{2}{c|}{ 0.19 $\pm$ 0.01}\\
11  & \textbf{     251.0} $\pm$ 72.97      &        4\,236 & \textbf{ 0.05} $\pm$  0.01  &     0.13 & \multicolumn{2}{c|}{ 0.28 $\pm$ 0.01}\\
12  & \textbf{     406.9} $\pm$ 74.07      &        7\,577 & \textbf{ 0.09} $\pm$  0.01  &     0.18 & \multicolumn{2}{c|}{ 0.41 $\pm$ 0.03}\\
13  & \textbf{     655.6} $\pm$ 170.80     &       20\,005 & \textbf{ 0.16} $\pm$  0.03  &     0.47 & \multicolumn{2}{c|}{ 0.71 $\pm$ 0.11}\\
14  & \textbf{  1\,506.3} $\pm$ 268.70     &       50\,101 & \textbf{ 0.35} $\pm$  0.04  &     1.28 & \multicolumn{2}{c|}{ 1.52 $\pm$ 0.35}\\
15  & \textbf{  2\,184.4} $\pm$ 411.05     &      185\,336 & \textbf{ 0.60} $\pm$  0.10  &     5.26 & \multicolumn{2}{c|}{ 2.49 $\pm$ 0.51}\\
16  & \textbf{  4\,992.4} $\pm$ 1\,089.81  &      499\,569 & \textbf{ 1.32} $\pm$  0.29  &    16.3  & \multicolumn{2}{c|}{ 4.33 $\pm$ 0.77}\\
17  & \textbf{  9\,809.6} $\pm$ 1\,122.03  &   4\,186\,427 & \textbf{ 2.96} $\pm$  0.30  &   155    & \multicolumn{2}{c|}{ 7.88 $\pm$ 1.33}\\
18  & \textbf{ 20\,612.6} $\pm$ 5\,574.21  &   9\,811\,533 & \textbf{ 6.86} $\pm$  1.73  &   466    & \multicolumn{2}{c|}{13.86 $\pm$ 2.22}\\
19  & \textbf{ 46\,469.4} $\pm$ 9\,524.13  &  37\,572\,347 & \textbf{16.46} $\pm$  3.23  &  1849    & \multicolumn{2}{c|}{20.25 $\pm$ 2.37}\\
20  & \textbf{106\,454.3} $\pm$ 30\,381.30 & 185\,321\,420 &         41.84  $\pm$  10.38 & 11299    & \multicolumn{2}{c|}{\textbf{30.21} $\pm$ 2.86}\\
21  &         225\,597.7  $\pm$ 88\,749.61 &               &        104.13  $\pm$  34.47 &          & \multicolumn{2}{c|}{\textbf{43.73} $\pm$ 4.71}\\
22  &         549\,486.9  $\pm$168\,302.47 &               &        271.45  $\pm$  68.50 &          & \multicolumn{2}{c|}{\textbf{65.95} $\pm$ 4.36}\\
23  &      1\,142\,782.7  $\pm$276\,823.52 &               &        629.74  $\pm$ 115.34 &          & \multicolumn{2}{c|}{\textbf{90.15} $\pm$ 5.97}\\
\hline
\caption{Results of evaluation on the first set of random graphs compared with results of \cite{Jaehn2013CPP}.
\label{tab:cpp_rand_1}}
\end{longtable}
\normalsize

As seen from the table, for this set of random graphs, our approach significantly outperformed~\cite{Jaehn2013CPP} on networks of all sizes.
While averages of our trivial estimates $\overline{Q}^{*}$ are pretty close to $\overline{Q}$, indicating that generated random instances were similar to those used by~\cite{Jaehn2013CPP}, our estimates of upper bounds were always more than $20\%$ closer to the optimal solution.
For the largest instances with $20$ nodes, our approach considered about $1,800$ times fewer nodes and completed more than $250$ times faster.
\added[id=A]{On the other hand, we can see that for larger instances CPLEX starts to perform even faster than the proposed method.}

Graphs in the second set were generated using a procedure that is supposed to resemble the process of creating similarity networks of~\cite{Grotschel1989cutting}.
First, for every graph with $n$ vertices, we fixed a parameter~$p$.
Then, for each vertex, we created a binary vector of length $p$, picking $0$ or $1$ with an equal probability of $0.5$.
Finally, the weight of the edge between vertices $i$ and $j$ was set to $p$ minus doubled the number of positions where vectors of $i$ and $j$ differ.
For every $n$ from $10$ to \replaced[id=A]{$30$}{$24$} and $p$ from the set $\{1, 2, 3, 5, 10, 50, 100\}$ we generated $5$ instances of random graphs.
We show results for this set in Table~\ref{tab:cpp_rand_2}.
As previously, in every experiment, results are summed up over $35$ instances, and we report the mean and standard deviation calculated over ten experiment runs.

\footnotesize
\begin{longtable}{| r | c | c | c | c | c | c |}
\hline
n & $\overline{Q}^{*}$ & $\overline{Q}$ & $Q_{min}^{*}$ & $Q_{min}$ & $Q_{max}^{*}$ & $Q_{max}$ \\
\hline
\endfirsthead
\hline
\endhead
\hline
\endfoot

10 &  1.397 $\pm$ 0.030 &  1.387 & \textbf{0.999} $\pm$ 0.001 &         0.985  & \textbf{1.005} $\pm$ 0.002 &	1.122 \\
11 &  1.415 $\pm$ 0.021 &  1.476 & \textbf{0.999} $\pm$ 0.001 &         0.995  & \textbf{1.003} $\pm$ 0.001 &	1.177 \\
12 &  1.466 $\pm$ 0.035 &  1.421 &         0.999  $\pm$ 0.001 & \textbf{1.000} & \textbf{1.007} $\pm$ 0.003 &	1.149 \\
13 &  1.479 $\pm$ 0.020 &  1.437 & \textbf{0.998} $\pm$ 0.001 &         0.997  & \textbf{1.006} $\pm$ 0.004 &	1.144 \\
14 &  1.506 $\pm$ 0.022 &  1.516 & \textbf{0.998} $\pm$ 0.002 &         0.991  & \textbf{1.008} $\pm$ 0.002 &	1.173 \\
15 &  1.526 $\pm$ 0.018 &  1.546 & \textbf{0.998} $\pm$ 0.001 &         0.995  & \textbf{1.010} $\pm$ 0.002 &	1.178 \\
16 &  1.554 $\pm$ 0.022 &  1.541 & \textbf{0.998} $\pm$ 0.002 &         0.992  & \textbf{1.010} $\pm$ 0.003 &	1.181 \\
17 &  1.568 $\pm$ 0.030 &  1.569 & \textbf{0.997} $\pm$ 0.003 &         0.988  & \textbf{1.012} $\pm$ 0.004 &	1.188 \\
18 &  1.594 $\pm$ 0.024 &  1.575 & \textbf{0.997} $\pm$ 0.002 &         0.992  & \textbf{1.014} $\pm$ 0.003 &	1.195 \\
19 &  1.619 $\pm$ 0.017 &  1.592 & \textbf{0.998} $\pm$ 0.001 &         0.987  & \textbf{1.019} $\pm$ 0.006 &	1.214 \\
20 &  1.631 $\pm$ 0.024 &  1.630 & \textbf{0.998} $\pm$ 0.002 &         0.986  & \textbf{1.017} $\pm$ 0.006 &	1.228 \\
21 &  1.644 $\pm$ 0.019 &  1.631 & \textbf{0.996} $\pm$ 0.002 &         0.983  & \textbf{1.019} $\pm$ 0.004 &	1.229 \\
22 &  1.665 $\pm$ 0.029 &  1.639 & \textbf{0.996} $\pm$ 0.003 &         0.990  & \textbf{1.026} $\pm$ 0.008 &	1.232 \\
23 &  1.668 $\pm$ 0.021 &  1.632 & \textbf{0.997} $\pm$ 0.002 &         0.992  & \textbf{1.024} $\pm$ 0.006 &	1.218 \\
24 &  1.702 $\pm$ 0.033 &  1.728 & \textbf{0.996} $\pm$ 0.002 &         0.984  & \textbf{1.033} $\pm$ 0.009 &	1.269 \\
25 &  1.718 $\pm$ 0.015 &        &         0.996  $\pm$ 0.002 &                &         1.035  $\pm$ 0.005 &         \\
26 &  1.717 $\pm$ 0.028 &        &         0.996  $\pm$ 0.002 &                &         1.034  $\pm$ 0.010 &         \\
27 &  1.748 $\pm$ 0.031 &        &         0.995  $\pm$ 0.003 &                &         1.044  $\pm$ 0.012 &         \\
28 &  1.767 $\pm$ 0.023 &        &         0.997  $\pm$ 0.001 &                &         1.051  $\pm$ 0.009 &         \\
29 &  1.766 $\pm$ 0.026 &        &         0.996  $\pm$ 0.002 &                &         1.049  $\pm$ 0.009 &         \\
30 &  1.778 $\pm$ 0.020 &        &         0.996  $\pm$ 0.002 &                &         1.055  $\pm$ 0.007 &         \\

\hline
n & Nodes$^{*}$ & Nodes & $t^*$ (s) & $t$\,(s) & \multicolumn{2}{c|}{$t^{CPLEX}$\,(s)} \\
\hline
10 &	\textbf{     69.1} $\pm$       37.60 &	          488 &	\textbf{ 0.01} $\pm$   0.00 &       0.04  & \multicolumn{2}{c|}{ 0.17 $\pm$ 0.01}\\
11 &	\textbf{     70.9} $\pm$       34.78 &	          962 &	\textbf{ 0.02} $\pm$   0.00 &       0.05  & \multicolumn{2}{c|}{ 0.22 $\pm$ 0.00}\\
12 &	\textbf{    117.9} $\pm$       36.40 &	          972 &	\textbf{ 0.03} $\pm$   0.00 &       0.05  & \multicolumn{2}{c|}{ 0.32 $\pm$ 0.01}\\
13 &	\textbf{    169.9} $\pm$       61.57 &	       2\,178 &	\textbf{ 0.05} $\pm$   0.01 &       0.08  & \multicolumn{2}{c|}{ 0.41 $\pm$ 0.02}\\
14 &	\textbf{    255.8} $\pm$       68.65 &	       6\,158 &	\textbf{ 0.08} $\pm$   0.01 &       0.19  & \multicolumn{2}{c|}{ 0.54 $\pm$ 0.03}\\
15 &	\textbf{    310.1} $\pm$       62.96 &	       7\,819 &	\textbf{ 0.11} $\pm$   0.01 &       0.22  & \multicolumn{2}{c|}{ 0.76 $\pm$ 0.06}\\
16 &	\textbf{    456.3} $\pm$      174.42 &	      21\,752 &	\textbf{ 0.18} $\pm$   0.05 &       0.71  & \multicolumn{2}{c|}{ 1.03 $\pm$ 0.18}\\
17 &	\textbf{    657.3} $\pm$      298.50 &	     138\,305 &	\textbf{ 0.28} $\pm$   0.08 &       5.08  & \multicolumn{2}{c|}{ 1.44 $\pm$ 0.38}\\
18 &	\textbf{    835.1} $\pm$      346.38 &	     160\,195 &	\textbf{ 0.39} $\pm$   0.09 &       6.52  & \multicolumn{2}{c|}{ 2.17 $\pm$ 0.68}\\
19 &	\textbf{ 1\,758.0} $\pm$      822.40 &	  1\,389\,759 &	\textbf{ 0.82} $\pm$   0.28 &      66.4   & \multicolumn{2}{c|}{ 3.27 $\pm$ 0.46}\\
20 &	\textbf{ 1\,614.3} $\pm$      660.14 &	  2\,598\,775 &	\textbf{ 0.96} $\pm$   0.27 &     136     & \multicolumn{2}{c|}{ 4.40 $\pm$ 1.35}\\
21 &	\textbf{ 2\,683.3} $\pm$   1\,058.49 &	 11\,977\,231 &	\textbf{ 1.66} $\pm$   0.44 &     741     & \multicolumn{2}{c|}{ 6.10 $\pm$ 1.17}\\
22 &	\textbf{ 5\,235.1} $\pm$   3\,184.27 &	 14\,413\,288 &	\textbf{ 3.20} $\pm$   1.42 &     962     & \multicolumn{2}{c|}{ 8.59 $\pm$ 1.42}\\
23 &	\textbf{ 7\,615.7} $\pm$   4\,868.18 &	 25\,313\,750 &	\textbf{ 5.20} $\pm$   2.95 &  1\,805     & \multicolumn{2}{c|}{11.82 $\pm$ 2.64}\\
24 &	\textbf{14\,655.1} $\pm$   8\,025.45 &	778\,958\,420 &	\textbf{10.22} $\pm$   4.80 & 67\,034     & \multicolumn{2}{c|}{18.64 $\pm$ 3.54}\\
25 &            25\,972.3  $\pm$  12\,628.10 &                &         20.35  $\pm$   8.16 &             & \multicolumn{2}{c|}{\textbf{24.78} $\pm$ 4.47}\\
26 &            38\,526.0  $\pm$  26\,307.77 &                &         35.43  $\pm$  23.49 &             & \multicolumn{2}{c|}{\textbf{32.32} $\pm$ 9.02}\\
27 &           130\,140.0  $\pm$ 116\,007.61 &                &        123.46  $\pm$ 100.52 &             & \multicolumn{2}{c|}{\textbf{50.19} $\pm$ 14.02}\\
28 &           152\,854.7  $\pm$  63\,692.59 &                &        159.94  $\pm$  58.73 &             & \multicolumn{2}{c|}{\textbf{63.77} $\pm$ 12.06}\\
29 &           234\,440.0  $\pm$ 136\,880.69 &                &        275.85  $\pm$ 161.35 &             & \multicolumn{2}{c|}{\textbf{81.93} $\pm$ 20.81}\\
30 &           640\,730.4  $\pm$ 392\,377.98 &                &        767.49  $\pm$ 435.67 &             & \multicolumn{2}{c|}{\textbf{115.10} $\pm$ 20.48}\\

\hline
\caption{Results of evaluation on the second set of random graphs compared with results of \cite{Jaehn2013CPP}.
\label{tab:cpp_rand_2}}
\end{longtable}
\normalsize

Again, we can see that our approach gave a very significant speedup in execution time compared to the method of~\cite{Jaehn2013CPP}.
Comparable execution time for smaller instances could be explained by the simplicity of these networks, where even straightforward but fast methods work well.
\added[id=A]{We can see that again, similar to results of~\cite{Jaehn2013CPP}, the largest instances are faster solved by CPLEX. That can suggest that for now our method is better suited for smaller networks. In our future research on this problem, we will try to address this by applying column and row generation methods to speed up the solution of underlying LP problems.}

The third set consists of graphs created using the same procedure as for the first set, but then the weight of each edge was set to zero with $40\%$ probability in the first subset (Table~\ref{tab:cpp_rand_31}) and $80\%$ in the second subset (Table~\ref{tab:cpp_rand_32}).

\footnotesize
\begin{longtable}{| r | c | c | c | c | c | c |}
\hline
n & $\overline{Q}^{*}$ & $\overline{Q}$ & $Q_{min}^{*}$ & $Q_{min}$ & $Q_{max}^{*}$ & $Q_{max}$ \\
\hline
\endfirsthead
\hline
\endhead
\hline
\endfoot

10  & 1.421 $\pm$ 0.066 & 1.468 & \textbf{0.999} $\pm$ 0.002 & 0.987 & \textbf{1.007} $\pm$ 0.007 & 1.153 \\
11  & 1.420 $\pm$ 0.041 & 1.494 & \textbf{0.998} $\pm$ 0.003 & 0.985 & \textbf{1.007} $\pm$ 0.004 & 1.173 \\
12  & 1.486 $\pm$ 0.043 & 1.498 & \textbf{0.997} $\pm$ 0.002 & 0.983 & \textbf{1.008} $\pm$ 0.005 & 1.167 \\
13  & 1.541 $\pm$ 0.054 & 1.513 & \textbf{0.998} $\pm$ 0.003 & 0.988 & \textbf{1.012} $\pm$ 0.007 & 1.192 \\
14  & 1.582 $\pm$ 0.026 & 1.492 & \textbf{0.998} $\pm$ 0.003 & 0.990 & \textbf{1.014} $\pm$ 0.007 & 1.184 \\
15  & 1.627 $\pm$ 0.037 & 1.616 & \textbf{0.995} $\pm$ 0.003 & 0.983 & \textbf{1.018} $\pm$ 0.004 & 1.243 \\
16  & 1.648 $\pm$ 0.028 & 1.696 & \textbf{0.995} $\pm$ 0.002 & 0.986 & \textbf{1.019} $\pm$ 0.006 & 1.267 \\
17  & 1.707 $\pm$ 0.034 & 1.750 & \textbf{0.995} $\pm$ 0.002 & 0.975 & \textbf{1.020} $\pm$ 0.006 & 1.307 \\
18  & 1.747 $\pm$ 0.023 & 1.699 & \textbf{0.995} $\pm$ 0.003 & 0.974 & \textbf{1.028} $\pm$ 0.007 & 1.263 \\
19  & 1.768 $\pm$ 0.027 & 1.800 & \textbf{0.996} $\pm$ 0.002 & 0.984 & \textbf{1.028} $\pm$ 0.006 & 1.315 \\
20  & 1.816 $\pm$ 0.035 & 1.850 & \textbf{0.995} $\pm$ 0.004 & 0.985 & \textbf{1.038} $\pm$ 0.009 & 1.326 \\
21  & 1.842 $\pm$ 0.038 &       &         0.994  $\pm$ 0.003 &       &         1.045  $\pm$ 0.009 &       \\
22  & 1.863 $\pm$ 0.028 &       &         0.993  $\pm$ 0.002 &       &         1.046  $\pm$ 0.010 &       \\
23  & 1.897 $\pm$ 0.034 &       &         0.993  $\pm$ 0.003 &       &         1.053  $\pm$ 0.010 &       \\
\hline
n & Nodes$^{*}$ & Nodes & $t^*$ (s) & $t$\,(s) & \multicolumn{2}{c|}{$t^{CPLEX}$\,(s)} \\
\hline
10 &     33.2 $\pm$      18.30 &         388 &         0.01  $\pm$ 0.00 &   0.01 & \multicolumn{2}{c|}{0.19 $\pm$ 0.01} \\
11 &     60.6 $\pm$      27.95 &         810 &         0.01  $\pm$ 0.00 &   0.01 & \multicolumn{2}{c|}{0.26 $\pm$ 0.01} \\
12 &    100.0 $\pm$      47.25 &      2\,442 & \textbf{0.02} $\pm$ 0.00 &   0.04 & \multicolumn{2}{c|}{0.33 $\pm$ 0.01} \\
13 &    183.0 $\pm$      68.13 &      5\,128 & \textbf{0.03} $\pm$ 0.01 &   0.09 & \multicolumn{2}{c|}{0.46 $\pm$ 0.03} \\
14 &    264.8 $\pm$      91.02 &      4\,836 & \textbf{0.05} $\pm$ 0.01 &   0.08 & \multicolumn{2}{c|}{0.64 $\pm$ 0.05} \\
15 &    465.1 $\pm$      96.36 &     25\,647 & \textbf{0.08} $\pm$ 0.01 &   0.54 & \multicolumn{2}{c|}{0.90 $\pm$ 0.09} \\
16 &    640.8 $\pm$     190.18 &     54\,728 & \textbf{0.13} $\pm$ 0.03 &   1.38 & \multicolumn{2}{c|}{1.42 $\pm$ 0.24} \\
17 & 1\,134.5 $\pm$     359.41 &    140\,765 & \textbf{0.22} $\pm$ 0.05 &   3.93 & \multicolumn{2}{c|}{2.11 $\pm$ 0.40} \\
18 & 1\,695.5 $\pm$     406.83 &    382\,507 & \textbf{0.36} $\pm$ 0.08 &  11.59 & \multicolumn{2}{c|}{3.30 $\pm$ 0.79} \\
19 & 2\,752.3 $\pm$     592.63 & 1\,469\,527 & \textbf{0.63} $\pm$ 0.10 &  55.69 & \multicolumn{2}{c|}{5.42 $\pm$ 1.06} \\
20 & 3\,823.5 $\pm$     950.63 & 3\,195\,924 & \textbf{1.04} $\pm$ 0.24 & 114    & \multicolumn{2}{c|}{9.78 $\pm$ 2.09} \\
21 & 7\,842.5 $\pm$  2\,975.48 &             & \textbf{ 2.32}$\pm$ 0.74 &        & \multicolumn{2}{c|}{16.91 $\pm$ 3.17}\\
22 &11\,958.9 $\pm$  3\,771.55 &             & \textbf{ 4.01}$\pm$ 1.20 &        & \multicolumn{2}{c|}{24.19 $\pm$ 4.57}\\
23 &20\,967.8 $\pm$  4\,610.94 &             & \textbf{ 7.83}$\pm$ 1.51 &        & \multicolumn{2}{c|}{38.51 $\pm$ 3.42}\\
\hline
\caption{Results of evaluation on the third set of random graphs with $40\%$ probability of edge weight being set to zero compared with results of \cite{Jaehn2013CPP}.
\label{tab:cpp_rand_31}}
\end{longtable}

\begin{longtable}{| r | c | c | c | c | c | c |}
\hline
n & $\overline{Q}^{*}$ & $\overline{Q}$ & $Q_{min}^{*}$ & $Q_{min}$ & $Q_{max}^{*}$ & $Q_{max}$ \\
\hline
\endfirsthead
\hline
\endhead
\hline
\endfoot

10 & 1.037 $\pm$ 0.019 & 1.060 & \textbf{1.000} $\pm$ 0.000  & 0.992 & \textbf{1.000} $\pm$ 0.000 &	1.037 \\
11 & 1.065 $\pm$ 0.036 & 1.067 &         1.000  $\pm$ 0.000  & 1.000 & \textbf{1.000} $\pm$ 0.000 &	1.013 \\
12 & 1.064 $\pm$ 0.024 & 1.050 & \textbf{1.000} $\pm$ 0.001  & 0.995 & \textbf{1.000} $\pm$ 0.000 &	1.006 \\
13 & 1.096 $\pm$ 0.027 & 1.088 & \textbf{1.000} $\pm$ 0.000  & 0.977 & \textbf{1.001} $\pm$ 0.002 &	1.023 \\
14 & 1.091 $\pm$ 0.015 & 1.048 & \textbf{1.000} $\pm$ 0.000  & 0.999 & \textbf{1.001} $\pm$ 0.002 &	1.019 \\
15 & 1.102 $\pm$ 0.021 & 1.110 & \textbf{1.000} $\pm$ 0.001  & 0.982 & \textbf{1.003} $\pm$ 0.005 &	1.059 \\
16 & 1.124 $\pm$ 0.026 & 1.135 & \textbf{1.000} $\pm$ 0.001  & 0.979 & \textbf{1.004} $\pm$ 0.004 &	1.062 \\
17 & 1.131 $\pm$ 0.021 & 1.114 & \textbf{1.000} $\pm$ 0.000  & 0.989 & \textbf{1.005} $\pm$ 0.006 &	1.080 \\
18 & 1.156 $\pm$ 0.027 & 1.143 & \textbf{1.000} $\pm$ 0.003  & 0.985 & \textbf{1.006} $\pm$ 0.005 &	1.082 \\
19 & 1.160 $\pm$ 0.015 & 1.195 & \textbf{0.999} $\pm$ 0.001  & 0.985 & \textbf{1.010} $\pm$ 0.006 &	1.108 \\
20 & 1.175 $\pm$ 0.027 & 1.147 & \textbf{0.999} $\pm$ 0.002  & 0.986 & \textbf{1.019} $\pm$ 0.012 &	1.072 \\
21 & 1.194 $\pm$ 0.031 &       &         0.998  $\pm$ 0.003  &       &         1.016  $\pm$ 0.006 &       \\
22 & 1.213 $\pm$ 0.026 &       &         0.998  $\pm$ 0.002  &       &         1.027  $\pm$ 0.008 &       \\
23 & 1.234 $\pm$ 0.017 &       &         0.998  $\pm$ 0.002  &       &         1.032  $\pm$ 0.010 &       \\
\hline
n & Nodes$^{*}$ & Nodes & $t^*$ (s) & $t$\,(s) & \multicolumn{2}{c|}{$t^{CPLEX}$\,(s)} \\
\hline
10 & \textbf{   0} $\pm$     0 &     30 &         0.00  $\pm$ 0.00 & 0    & \multicolumn{2}{c|}{0.18 $\pm$ 0.01} \\
11 & \textbf{ 0.8} $\pm$  2.53 &     14 &         0.00  $\pm$ 0.00 & 0    & \multicolumn{2}{c|}{0.24 $\pm$ 0.00} \\
12 & \textbf{ 1.3} $\pm$  3.77 &     55 &         0.00  $\pm$ 0.00 & 0    & \multicolumn{2}{c|}{0.31 $\pm$ 0.00} \\
13 & \textbf{ 4.4} $\pm$  6.10 &    167 & \textbf{0.00} $\pm$ 0.00 & 0.01 & \multicolumn{2}{c|}{0.41 $\pm$ 0.00} \\
14 & \textbf{ 4.3} $\pm$  5.44 &     71 &         0.00  $\pm$ 0.00 & 0    & \multicolumn{2}{c|}{0.52 $\pm$ 0.02} \\
15 & \textbf{ 8.5} $\pm$ 12.83 &    472 & \textbf{0.00} $\pm$ 0.00 & 0.01 & \multicolumn{2}{c|}{0.64 $\pm$ 0.01} \\
16 & \textbf{11.7} $\pm$ 10.81 &    720 & \textbf{0.00} $\pm$ 0.00 & 0.01 & \multicolumn{2}{c|}{0.78 $\pm$ 0.01} \\
17 & \textbf{20.7} $\pm$ 15.24 &    789 & \textbf{0.00} $\pm$ 0.00 & 0.02 & \multicolumn{2}{c|}{0.96 $\pm$ 0.02} \\
18 & \textbf{52.3} $\pm$ 47.84 &    979 & \textbf{0.01} $\pm$ 0.00 & 0.02 & \multicolumn{2}{c|}{1.15 $\pm$ 0.02} \\
19 & \textbf{56.2} $\pm$ 25.47 & 2\,601 & \textbf{0.01} $\pm$ 0.00 & 0.06 & \multicolumn{2}{c|}{1.35 $\pm$ 0.01} \\
20 & \textbf{96.9} $\pm$ 46.48 & 2\,423 & \textbf{0.01} $\pm$ 0.00 & 0.05 & \multicolumn{2}{c|}{1.60 $\pm$ 0.04} \\
21 &        133.7  $\pm$ 57.00 &        & \textbf{0.02} $\pm$ 0.00 &      & \multicolumn{2}{c|}{1.86 $\pm$ 0.03} \\
22 &        233.8  $\pm$ 73.73 &        & \textbf{0.03} $\pm$ 0.01 &      & \multicolumn{2}{c|}{2.20 $\pm$ 0.05} \\
23 &        295.4  $\pm$103.47 &        & \textbf{0.04} $\pm$ 0.01 &      & \multicolumn{2}{c|}{2.58 $\pm$ 0.10} \\
\hline
\caption{Results of evaluation on the third set of random graphs with $80\%$ probability of edge weight being set to zero compared with results of \cite{Jaehn2013CPP}.
\label{tab:cpp_rand_32}}
\end{longtable}
\normalsize

These subsets appeared to be the easiest to solve.
Here again, our method was faster than its competitor, but by a smaller margin primarily because, for such an easy set, there was little room for improvement.
While the simplicity of this set for both methods is surprising because the abundance of zero-weight edges means fewer triangles and chains, our results confirm the conclusion of~\cite{Jaehn2013CPP} that zeroing out edges at random only makes instances easier to solve.

As we mentioned in section~\ref{sect:b&b}, we used only chains to estimate upper bounds in our algorithms.
However, as we proved in corollary~\ref{cor:RelaxationVsChains}, this approach cannot provide an upper bound tighter than the solution of the relaxed problem~(\ref{eq:ILP}).
Therefore, to show how our method could be extended, we did an experiment where we also used stars to estimate upper bounds.
We applied our algorithm to maximize the modularity of two well-studied real-world networks: a social network of frequent associations between 62 dolphins~\citep{Lusseau2003Dolphins} and a co-appearance network of characters in Les~Miserables novel~\citep{Knuth1993GraphBase}.
It took a couple of minutes to construct a set of stars and solve the initial LP problem.
However, for both networks, found solutions ($Q_{max}$) were already equal to the feasible solutions found by Combo, which proved their optimality, while the solutions of relaxed problem~(\ref{eq:ILP}) give higher values of upper bounds~\citep{Miyauchi2013}.
These results make us believe that further improvements to our algorithm will allow us to achieve even better performance for more difficult networks.

\section{Conclusions}
 
We propose a two-stage method, providing an efficient solution for the clique partitioning problem in some cases. First, we define penalizing subnetworks and use them to calculate the upper bounds of the clique quality function. In many cases, our method is much faster than other methods for upper-bounds estimation, and for many networks, it finds tighter upper bounds. Second, we present an algorithm that uses found upper bounds in the branch-and-bound technique to solve the problem exactly. Our experiments showed that the proposed algorithm drastically outperforms some previously known approaches even when using only a single class of penalizing subnetworks that we call chains. Moreover, the proposed heuristic, which allows finding upper bounds using chains quickly, works much faster than a well-known alternative approach leveraging a linear programming problem, while the resulting upper bounds are tight enough for many networks to find the exact solution efficiently.

We also provide a framework for using more general penalizing subnetworks when chains are not effective enough. E.g., we introduce another class of subnetworks called stars that can help find upper bounds much tighter than those found by chains and even by a linear programming-based method. Constructing more diverse sets of penalizing subnetworks and improving the efficiency of incorporating them into the method can further improve finding exact solutions to the clique partitioning problem.
\added[id=A]{Since some larger graph instances are still solved faster by standard packages like CPLEX, we plan to incorporate column and row generation methods and cutting plane techniques into our algorithm.} We believe that future work in this direction could provide efficient solutions for the clique partitioning and its particular case~--- a modularity maximization problem~--- for a broader range of networks, including larger ones.

\backmatter

\bmhead{Acknowledgments}
We thank Daniel Bretsko and Margarita Mishina for their careful review of the paper and constructive comments and suggestions.

This version of the article has been accepted for publication, after peer review but is not the Version of Record and does not reflect post-acceptance improvements, or any corrections. The Version of Record is available online at: \url{https://doi.org/10.1007/s00186-023-00835-y}.

\section*{Statements and Declarations}

\begin{itemize}
\item \textbf{Funding}:
The work of Alexander Belyi and Stanislav Sobolevsky was partially supported by the MUNI Award in Science and Humanities (MASH Belarus) of the Grant Agency of Masaryk University under the Digital City project (MUNI/J/0008/2021).
The work of Alexander Belyi was also partially supported by the National Research Foundation (prime minister’s office, Singapore), under its CREATE program, Singapore-MIT Alliance for Research and Technology (SMART) Future Urban Mobility (FM) IRG.
The work of Stanislav Sobolevsky was also partially supported by ERDF ``CyberSecurity, CyberCrime and Critical Information Infrastructures Center of Excellence'' (No. CZ.02.1.01/0.0/0.0/16\_019/0000822).
\item \textbf{Competing interests}:
The authors declare that they have no conflict of interest.
\item \textbf{Availability of code, data and materials}:
Code, data and materials are available at \url{https://github.com/Alexander-Belyi/best-partition}.
\item \textbf{Authors' contributions}:
All authors contributed to the study conception and design. Methods and algorithms were developed by AB and SS. Computational experiments and their analysis were performed by AB. The first draft of the manuscript was written by AB, and all authors commented on previous versions of the manuscript. All authors read and approved the final manuscript.
\end{itemize}

\begin{appendices}
\section{\label{sect:Appendix}}

\begin{theorem}\label{thrm:RelaxationForChains}
Solution of the relaxed problem~(\ref{eq:ILP}) (i.e., when constraints $x_{ij}\in\{0,1\}$ are replaced with $x_{ij}\in[0,1]$) for any chain (see definition~\ref{Def:Chain}) always finds upper bound equal to the maximum of $Q$ for this chain.
\end{theorem}
\begin{proof}
Consider a chain with $k$ nodes and edge weights $w_{1,2}, w_{2,3},\dots, w_{k-1,k}, -w_{1,k}$, where $w_{1,2}, w_{2,3}, \dots, w_{k-1,k}, w_{1,k} > 0$ (see Fig.~\ref{fig:ChainStar}A).
Let $w_{m,m+1} = \min_{i\in\{1..k-1\}}w_{i,i+1}$.
There are two cases:

1. $w_{1,k} < w_{m,m+1}$ \\
In this case, the optimal clique partition is to assign all nodes to one clique, i.e., to include the negative edge $\{1, k\}$,
$Q_{opt} = w_{1,2} + w_{2,3} +\dots+ w_{k-1,k} - w_{1,k}$.
So, we want to show that $w_{1,2} + w_{2,3} +\dots+ w_{k-1,k} - w_{1,k} \geq x_{1,2}w_{1,2} + x_{2,3}w_{2,3} +\dots+ x_{k-1,k}w_{k-1,k} - x_{1,k}w_{1,k}$, for all $x_{ij}\in[0,1]$ satisfying triangle inequalities of the problem~(\ref{eq:ILP}).
After regrouping, we get: $(1-x_{1,2})w_{1,2}+(1-x_{2,3})w_{2,3}+\dots+(1-x_{k-1,k})w_{k-1,k} \geq (1-x_{1,k})w_{1,k}$.
Since all $w_{i,i+1}>w_{1,k}>0$ and $0 \leq x_{i,j} \leq 1$, it is enough to show that
$1-x_{1,2}+1-x_{2,3}+\dots+1-x_{k-1,k} \geq 1-x_{1,k}$.\\
From the constraints of problem~(\ref{eq:ILP}), we have:\\
$1 \geq x_{1,2} + x_{2,k} - x_{1,k} \Leftrightarrow 1 - x_{1,2} \geq x_{2,k} - x_{1,k}$,\\
$1 \geq x_{2,3} + x_{3,k} - x_{2,k} \Leftrightarrow 1 - x_{2,3} \geq x_{3,k} - x_{2,k}$,\\
$\cdots$\\
$1 \geq x_{k-2,k-1} + x_{k-1,k} - x_{k-2,k} \Leftrightarrow 1 - x_{k-2,k-1} \geq x_{k-1,k} - x_{k-2,k}$.\\
After summing them, we get\\
$1-x_{1,2} + 1-x_{2,3} +\dots+ 1-x_{k-2,k-1} \geq x_{k-1,k} - x_{1,k}$,\\
and by adding $1$ to both sides and moving $x_{k-1,k}$ to the left, we get the needed inequality.

2. $w_{m,m+1} \leq w_{1,k}$ \\
In this case, an optimal solution is to split all nodes into two groups by excluding negative and the `cheapest' positive edges,
$Q_{max} = w_{1,2} + w_{2,3} +\dots+ w_{m-1,m} + w_{m+1,m+2} +\dots+ w_{k-1,k}$.
Now, we want to show that $w_{1,2} + w_{2,3} +\dots+ w_{m-1,m} + w_{m+1,m+2} +\dots+ w_{k-1,k} \geq x_{1,2}w_{1,2} + x_{2,3}w_{2,3} +\dots+ x_{k-1,k}w_{k-1,k} - x_{1,k}w_{1,k}$.
After rearranging, we have: $(1-x_{1,2})w_{1,2}+(1-x_{2,3})w_{2,3}+\dots+(1-x_{m-1,m})w_{m-1,m} + (1-x_{m+1,m+2})w_{m+1,m+2}+\dots+(1-x_{k-1,k})w_{k-1,k} + x_{1,k}w_{1,k} \geq x_{m,m+1}w_{m,m+1}$.
Since all $w_{i,i+1}\geq w_{m,m+1}>0$, $w_{1,k} \geq w_{m,m+1}$ and $0 \leq x_{i,j} \leq 1$, it is enough to show that
$1-x_{1,2} + 1-x_{2,3} +\dots+ 1-x_{m-1,m} + 1-x_{m+1,m+2} +\dots+ 1-x_{k-1,k} + x_{1,k} \geq x_{m,m+1}$.
But we have already shown that
$1-x_{1,2}+1-x_{2,3}+\dots+1-x_{k-1,k} \geq 1-x_{1,k}$,
so
$1-x_{1,2} + 1-x_{2,3} +\dots+ 1-x_{m-1,m} + 1-x_{m+1,m+2} +\dots+ 1-x_{k-1,k} + x_{1,k} \geq 1-(1-x_{m,m+1}) = x_{m,m+1}$.
\end{proof}

\begin{corollary}\label{cor:RelaxationVsChains}
Solution of the relaxed problem~(\ref{eq:ILP}) always finds an upper bound on $Q$ that is the same as or tighter than the upper bound $Q_{max} = \overline{Q}(G) - \sum_k{\lambda_k P_k}$ obtained by solving problem~(\ref{eq:LPproblem}) for some set of chains $\{S_k\}$ with penalties $\{P_k\}$.
\end{corollary}
\begin{proof}
Following the proof of Theorem~\ref{thrm:PenaltyTheorem}, network $G$ could be represented as sum $G = S^L + R$ of the linear combination of chains $S^L$ and some residual subnetwork $R$ with edge weights $w^{*R}_{ij} = w_{ij} - w^*_{ij}$,
and for any partition, $Q(G) = Q(S^L) + Q(R)$.
So, for the optimal solution $\{x_{ij}\}$ of the relaxed problem~(\ref{eq:ILP}), we have:
$
\sum_{i<j}{w_{ij}\cdot x_{ij}} = 
\sum_{i<j}{x_{ij}\cdot (w^{*R}_{ij} + w^*_{ij})} = 
\sum_{i<j}{x_{ij}\cdot w^{*R}_{ij} + \sum_{i<j}{x_{ij}\cdot\sum_k{\lambda_k w^{*k}_{ij}}}} \leq
\sum_{w^{*R}_{ij}>0}{w^{*R}_{ij}} + \sum_k{\lambda_k\cdot\sum_{i<j}{x_{ij}\cdot w^{*k}_{ij}}} \leq
\sum_{w^{*R}_{ij}>0}{w^{*R}_{ij}} + \sum_k{\lambda_k\cdot Q_{opt}(S^L)} =
\overline{Q}(R) + \sum_k{\lambda_k\cdot\left(\sum_{w^{*k}_{ij}>0}{w^{*k}_{ij}} - P_k\right)} =
\overline{Q}(R) + \overline{Q}(S^L) - \sum_k{\lambda_k P_k} =
\overline{Q}(G) - \sum_k{\lambda_k P_k}
$.
\end{proof}
\end{appendices}


\end{document}